\documentclass[12pt, draftclsnofoot, onecolumn]{IEEEtran}
\usepackage{amsmath}
\usepackage{amsthm}
\usepackage{exscale,relsize}
\usepackage{cite}
\usepackage{graphicx}
\usepackage{cases}
\usepackage{epstopdf}
\usepackage{amsfonts,amsmath,amssymb}
\usepackage{graphicx}
\usepackage{bm}
\usepackage{bbm}
\usepackage{color}
\usepackage{caption}
\usepackage{enumerate}
\DeclareMathOperator*{\st}{s.t.}
\theoremstyle{definition}
\newtheorem{Lemma}{Lemma}

\newtheorem{Assumption}{Assumption}
\newtheorem{prop}{Proposition}

\newtheorem{Remark}{Remark}

\usepackage{algorithm,algorithmic}

\ifCLASSINFOpdf
\else
\fi
\begin{document}
%
\title{\LARGE{Online Cognitive Data Sensing and Processing Optimization in Energy-harvesting Edge Computing Systems}

\author{Xian~Li,~
	Suzhi~Bi,~
	Zhi~Quan,~
	and~Hui Wang}

\thanks{{X.~Li, S.~Bi, and Z.~Quan are with the College of Electronics and Information Engineering, Shenzhen University, China
		({xianli,bsz,zquan}@szu.edu.cn). H.~Wang is with the Shenzhen Institute of Information Technology, China (wanghui@sziit.edu.cn).} Part of this work has been submitted to IEEE/CIC International Conference on Communications in China (ICCC) 2021 \cite{LiICCC2021}. }}

%
%
%

\maketitle

\vspace{-3em}
\begin{abstract}
Mobile edge computing (MEC) has recently become a prevailing technique to alleviate the intensive computation burden in Internet of Things (IoT) networks. However, the limited device battery capacity and stringent spectrum resource significantly restrict the data processing performance of MEC-enabled IoT networks. To address the two performance limitations, we consider in this paper an MEC-enabled IoT system with a wireless device (WD) replenishing its battery by means of energy harvesting (EH) and opportunistically accessing the licensed spectrum of an overlaid primary communication link to offload its sensing data to an MEC server (MS) for edge processing. Under time-varying fading channel, random energy arrivals, and stochastic ON-OFF state of the primary link, we aim to design an online algorithm to jointly control the cognitive data sensing rate and processing method (i.e., local and edge processing) without knowing future system information. In particular, we aim to maximize the long-term average sensing rate of the WD subject to quality of service (QoS) requirement of primary link, average power constraint of MS and data queue stability of both MS and WD. We formulate the problem as a multi-stage stochastic optimization and propose an online algorithm named PLySE that applies the perturbed Lyapunov optimization technique to decompose the original problem into per-slot deterministic optimization problems. For each per-slot problem, we derive the closed-form optimal solution of data sensing and processing control to facilitate low-complexity real-time implementation. Interestingly, our analysis finds that the optimal solution exhibits an threshold-based structure related to the current energy state, secondary queueing backlogs and primary link activity. Simulation results collaborate with our analysis and demonstrate more than 46.7\% data sensing rate improvement of the proposed PLySE over representative benchmark methods.
\end{abstract}
\vspace{-5pt}
\begin{IEEEkeywords}\vspace{-0em}
	Mobile edge computing, energy harvesting, cognitive radio system, online optimization algorithm.
\end{IEEEkeywords}

\section{Introduction}
\subsection{Motivations and Contributions}
Mobile edge computing (MEC) has been widely recognized as a key enabling technology towards data-intensive and latency-critic applications in Internet of Things (IoT) systems, such as smart manufacturing and industrial automation \cite{Aceto2019}, which deploys massive number of IoT devices (e.g., wireless sensors) capable of sensing, communication and computation. Via pushing the computation resource toward network edge, MEC allows IoT devices to offload intensive computation tasks to the nearby edge server for faster execution \cite{Mao2017,Bi2018a,You2016,Huang2020}. However, constrained by device size and manufacturing cost, an IoT device is often equipped with low-capacity battery that can hardly support sustainable operations especially in energy-hungry intelligent applications. Meanwhile, the large-scale deployment of IoT systems demands wideband spectrum resource and can cause severe interference to co-channel wireless communication systems, e.g., WiFi and cellular networks. Overall, the fundamental limitations on device energy and spectrum resource significantly restrict the data processing capability of existing MEC-enabled IoT systems.\

Recently, energy harvesting (EH) has emerged as a promising technique to mitigate the energy shortage of IoT devices \cite{Ulukus2015,Kamalinejad2015}. In particular, EH enables IoT devices to scavenge renewable energy from external sources like solar and thermal power to local batteries, with which they can replenish device power consumption in real-time. Recent studies have considered powering IoT devices in mobile edge computing applications by means of EH \cite{Xu2017a,Min2019,Zhang2018,Wu2018a,Mao2016,Chen2020}. Under stochastic renewable energy arrivals, the task offloading decisions of EH-enabled wireless devices (WDs) are highly coupled across sequential time slots by the time-varying battery state. Therefore, the optimal design of an online computation offloading strategy in EH-MEC system requires to reach a good balance between the current and future system performance.\ 

On another front, compared to purchasing expensive licensed bandwidth, it is more cost-effective to implement cognitive radio (CR) technique in IoT networks for opportunistic spectrum access to alleviate spectrum scarcity \cite{Quan2008,Wu2014,Khan2017}. In this case, IoT devices as the secondary users (SUs) periodically sense the spectrum licensed to primary users (PUs) and opportunistically access the channel given that the PUs can maintain a satisfied quality of service (QoS). Thanks to the effective spectral reuse of CR technology, we can achieve scalable deployment of IoT networks that coexist with existing communication infrastructures under stringent spectrum resource constraint.\

In this paper, we aim to build a sustainable and scalable MEC-enabled IoT system, where the IoT device relies on harvesting ambient renewable energy for power supply and uses CR technique to access the spectrum licensed to a primary communication system. The joint application of CR and EH techniques has the potential to fully address the two fundamental performance limitations of MEC-enabled IoT systems, however, also raises new technical challenges. On one hand, the optimal data sensing (collecting task data) and processing (local computing and task offloading) solution at each time instant is affected by the double randomness of PU link activity and energy arrivals. For instance, the offloading data rate is constrained by not only the available energy but also the current PU link ON/OFF operating state. On the other hand, the solutions are tightly coupled over time due to the temporal correlations of PU link activities and EH process. It can result in large task data queue backlog if we independently maximize the data processing rate within each time slot in a greedy manner. Overall, it requires jointly considering both the short-term and long-term effects of stochastic PU link state and EH process to optimize the data processing capability of an MEC-enabled IoT system. 
\begin{figure}[h]
	\centering
	\includegraphics[scale=0.85]{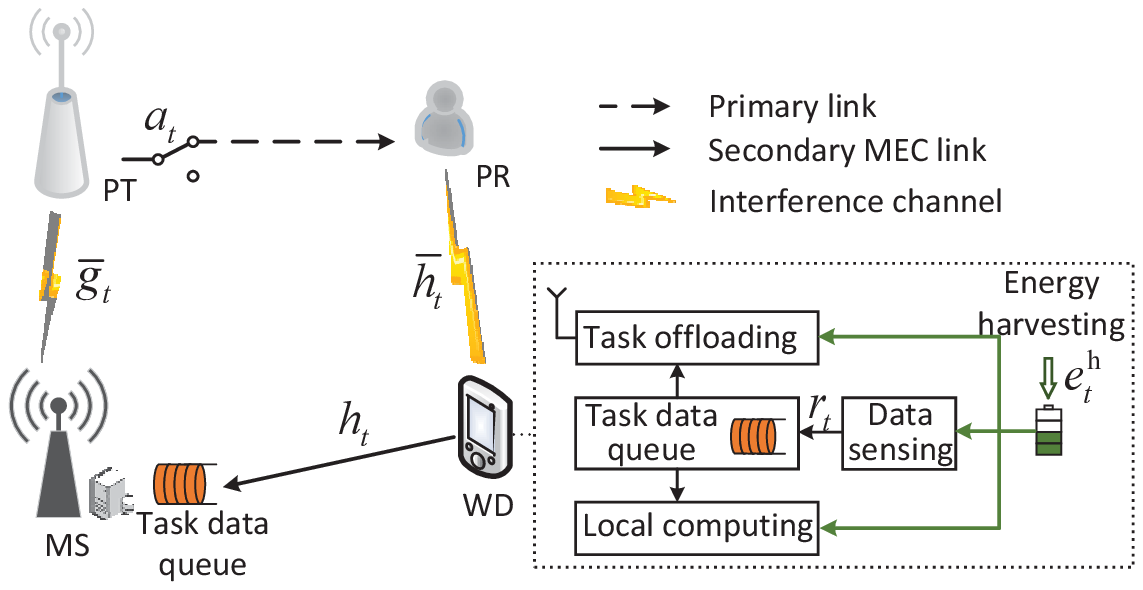}
	\captionsetup{font=footnotesize}
	\caption{The considered MEC system model.}
	\label{Sys_mod}
\end{figure}

As a starting point to gain essential insights of the optimal designs, we consider a basic system setup in a cognitive EH-MEC system that consists of a primary transmitter (PT), a primary receiver (PR), an MEC server (MS), and an energy-harvesting WD, as shown in Fig. \ref{Sys_mod}. Within each sequential time slot, the WD acquires sensing data from the environment and buffers the data in a task data queue. Meanwhile, it processes the task data either locally or by offloading to the MS utilizing the licensed spectrum of primary link. Under time-varying fading channel, random energy arrivals, and stochastic ON-OFF state of the primary link, our objective is to design an online algorithm that maximizes the average data sensing rate of the WD subject to QoS requirement of the primary link, long-term average power constraint at the MS, and data queue stability at both MS and WD. To the best of our knowledge, this is the first paper that studies the long-term performance optimization in a cognitive EH-MEC system. The main contributions of this paper are:
\begin{itemize}
	\item \emph{Joint Cognitive Data Sensing and Processing Design}: We formulate the problem as a multi-stage stochastic optimization that decides the cognitive data sensing and processing solutions in sequential time slots. In particular, we focus on designing a practical online algorithm that obtains the control solution in each time slot without future system information. The two major challenges of the online design are to satisfy all the long-term constraints under the randomness of multiple system parameters and the strong coupling of control solutions over different time slots. 	
	\item \emph{Low-complexity Online Algorithm}: We propose an online algorithm named PLySE that applies the perturbed Lyapunov optimization technique to remove the time dependency of control decisions and transforms the multi-stage stochastic problem into per-slot deterministic optimization problems. For each per-slot problem, we propose a low-complexity algorithm that obtains the optimal solutions in closed-form to facilitate real-time implementation. Interestingly, we find that both the optimal sensing and task execution solutions follow simple threshold-based structure that is directly related to the primary link activity, secondary battery state and queueing backlogs.
	\item \emph{Theoretical Performance Analysis}: We prove that if a mild condition on the battery capacity is satisfied, the online algorithm PLySE always produces a feasible solution to the original multi-stage stochastic optimization problem. Meanwhile, we prove that PLySE achieves an $\left[O(1/V),O(V)\right]$ tradeoff asymptotically between the sensing rate and processing delay, where $V$ is a tunable parameter that a larger $V$ leads to higher sensing rate, and vice versa. We conduct extensive simulations to verify the performance of the proposed PLySE method, where we show that it satisfies all the long-term constraints and achieves more than 46.7\% higher sensing rate than the considered representative benchmark methods.
\end{itemize}


\subsection{Related Works}\label{sec2}
1) \textbf{Energy harvesting edge computing}:
EH-powered MEC technology has attracted significant attention in recent years. For example, \cite{Xu2017a} studied the optimal online offloading and autoscaling policy in a renewable edge computing system, where multiple co-located MSs are powered by a common EH cell site. \cite{Min2019} designed a deep reinforcement learning (DRL)-based online task offloading strategy for an EH IoT device assisted by multiple MSs. Considering an EH-MEC with an MS and an energy harvesting WD, \cite{Zhang2018} proposed a task offloading policy that balances execution delay and energy consumption of the WD. On the other hand, \cite{Wu2018a} focused on the geographical load balancing control among energy harvesting MSs, and proposed an online task offloading policy to minimize the backhaul data traffic and computation workload. Aiming at minimizing the execution latency and task failure rate, \cite{Mao2016} developed a Lyapunov optimization-based algorithm to jointly optimize the offloading decision and resource allocation in a single-WD EH-MEC system. Besides, considering hybrid energy supply at the edge devices, \cite{Chen2020} studied dynamic control on task scheduling and energy management to maximize the MEC system utility. It is worth mentioning that all these studies assume that the MEC system occupies a dedicated bandwidth. However, this will incur enormous spectrum license cost as the number of IoT devices soars up. In comparison, applying CR technology in MEC systems to opportunistically utilize the spectrum licensed to existing communication networks is a more cost-effective and scalable solution. In this case, the optimal task offloading solution is directly affected by the dynamic spectrum access method, where the conventional optimization methods (such as in \cite{Xu2017a,Min2019,Zhang2018,Wu2018a,Mao2016,Chen2020}) are no longer applicable.\

2) \textbf{Cognitive radio technology in MEC}:
Recently, some studies have introduced CR to MEC-enabled IoT systems to address the inherent spectrum scarcity problem. For example, \cite{Liu2019} maximized the energy efficiency of a secondary relay in a cognitive MEC network powered by wireless energy transfer. \cite{Apostolopoulos2020} designed a risk-aware task offloading strategy for secondary WDs based on game theory. \cite{Zhang2020a} studied the trade-off between spectrum sensing and task offloading, and proposed an energy-aware task offloading policy using deep reinforcement learning. To optimize the task processing performance, \cite{Si2017} focused on the routing and bandwidth scheduling in a three-layer cognitive MEC network. \cite{Du2018} considered a vehicular MEC system opportunistically accessing the TV white space for computation offloading and minimized the monetary cost of vehicular terminals and MEC server on task processing. Nonetheless, all these works either assume static wireless channel gain \cite{Liu2019,Apostolopoulos2020} or constant power supply at WDs \cite{Zhang2020a,Si2017,Du2018}. In the case of random channel fading and energy arrivals in our considered setup, the design of computation offloading strategy in the secondary MEC system is challenged by the unknown system parameters and the strong coupling of offloading decisions over different time slots. 

Overall, most existing works separately investigate the applications of EH and CR to MEC-enabled IoT network. In our considered cognitive EH-MEC system, however, the time-varying PU link operating state and EH process jointly affect the optimal data sensing and processing performance. In general, this calls for a joint consideration of random energy arrivals and stochastic ON-OFF state of the primary link in the online algorithm design. In this paper, we find that the optimal online data sensing and processing solution exhibits an interesting threshold structure directly related to the PU link activity and EH process, and validate the effectiveness of the online control via simulations.  

The rest of the paper is organized as follows. In Section \ref{sec3}, we introduce the system model of the cognitive EH-MEC network and formulate the sensing rate maximization problem. We propose the PLySE method to solve the problem in Section \ref{sec4} and show the feasibility and optimality of PLySE in Section \ref{secV}. In Section \ref{sec6}, we conduct numerical simulations to evaluate the proposed method. Finally, we conclude the paper in Section \ref{sec7}. 

\section{System Model and Problem Formulation}\label{sec3}
We consider in Fig. \ref{Sys_mod} an MEC system consisting of one MS and one energy harvesting WD that share the narrow spectrum band with a pair of primary users, i.e., a PT and a PR. The MS is connected to a stable power grid, while the WD is solely powered by external energy sources (e.g., solar power, thermal energy and wind energy). The MS assists the computation of WD in sequential time slots of equal duration $T$. Notice that the considered single-user MEC system may correspond to a typical user in a multi-user network, where the MS uses dedicated computing power, memory and communication bandwidth to serve the user via infrastructure vitalization technique \cite{Bi2020}. In time slot $t$, WD collects raw sensing data from the monitored environment, stores it into the local task data queue, and processes the data later. We assume that WD adopts a partial computation offloading rule, whereby the raw task data can be arbitrary divided into two parts with one computed locally and the other opportunistically offloaded to the MS for edge processing, using the spectrum licensed to the primary link. The primary link activity follows a general ON/OFF random process, e.g., 1-stage Markov process. Specifically, we use a binary indicator $a_t$ to denote the ON/OFF state of primary link, where $a_t=1$ and 0 denote that the primary link is active and idle, respectively. In particular, we assume that the PT is active with probability $\bar{a}$, i.e., $Pr(a_t=1)=\bar{a}$. With the primary and secondary systems coexisting in the same spectrum, computation offloading of secondary MEC system will cause interference to the concurrent primary communications.  \

\subsection{Communication Model}
Denote the channel gain between MS and WD, PT and MS, PR and WD as $h_t$, $\bar{g}_t$, and $\bar{h}_t$, respectively. We assume block fading for signal propagation, i.e., $h_t$, $\bar{g}_t$ and $\bar{h}_t$ are constant in a time slot $t$ while vary randomly from one slot to another. In order to maintain the service quality of primary transmission, we consider interference constraints which confines the co-channel interference and noise power suffered by the PU to a threshold $\Gamma_{\rm th}$, i.e.,
\begin{equation}\label{SNR_thre}
a_t\left(W\delta_{\rm p}^2+p_{t}^{\rm u}\bar{h}_t-\Gamma_{\rm th}\right)\leq 0, \forall t=0, 1, \cdots,
\end{equation}
where $W$ is the system bandwidth and $\delta_{\rm p}^2$ is the power spectrum density of additive while Gaussian noise (AWGN) at the PU. Notice that constraint \eqref{SNR_thre} is imposed only when the primary link is active ($a_t = 1$) and becomes immaterial otherwise (i.e., when $a_t = 0$). $p_t^{\rm u}$ is the transmit power of WD that is limited by the maximum value $p_{\rm max}$. We assume that the PT transmits with fixed power $P_{\rm B}$. Then, the offloading task data size from WD to MS in time slot $t$ is
\begin{equation}
l_t^{\rm off} = WT\log_2\left(1+p_{t}^{\rm u}\gamma_t\right),
\end{equation}
where $\gamma_t={h_{t}}/{\left(a_tP_{\rm B}\bar{g}_{t}+W\delta_{\rm s}^2\right)}$ is the signal-to-interference-plus-noise ratio (SINR) at the MS and $\delta_{\rm s}^2$ is the power spectrum density of AWGN at the MS. Correspondingly, the energy consumption of the WD for data transmission is
\begin{equation}
e_t^{\rm off} = p_t^{\rm u}T.
\end{equation}

\subsection{Task Data Sensing and Computation Model}
In time slot $t$, the WD collects $r_t\leq r_{\rm max}$ bits of raw measurement data, where $r_{\rm max}$ is the maximum sensing data size in $T$ (e.g., determined by the maximum measurement sampling rate or sensing resolution). We model the energy consumption on data sensing as\cite{Liu2020}
\begin{equation}
e_t^{\rm col}=e_{\rm unit}^{\rm col}r_t,
\end{equation}
where $e_{\rm unit}^{\rm col}$ in Jolue/bit is the unit energy cost for data sensing. The sensed data is piled in a local storage for subsequent task computation, where each task data bit is executed either locally at the WD or remotely at the MS via task offloading. We denote $f_t^{\rm u}$ as the local CPU frequency, where $f_t^{\rm u}$ is constrained by its maximum value $f^{\rm u}_{\rm max}$. Then, the local processing data size and the corresponding energy consumption for local computing are
\begin{equation}\label{eq_time_cost_comp_loc} 
l_t^{\rm loc} = f_t^{\rm u}T/C, ~e_t^{\rm loc} = \kappa_{\rm c}\left(f_t^{\rm u}\right)^2Cl_t^{\rm loc}=\kappa_{\rm c}\left(f_t^{\rm u}\right)^3T,
\end{equation}
respectively, where $C$ is the required CPU cycles to process one bit of data and $\kappa_{\rm c}$ is the energy efficiency for local computing. Similarly, we denote the edge CPU frequency at the MS as $f_t^{\rm s}$, which is upper bounded by $f^{\rm s}_{\rm max}$. The edge computation data size $l_t^{\rm edg}$ and the corresponding energy consumption $e_t^{\rm edg}$ are
\begin{equation}\label{eq_time_cost_comp_edg} 
l_t^{\rm edg} = f_t^{\rm s}T/C, ~e_t^{\rm edg} = \kappa_{\rm e}\left(f_t^{\rm s}\right)^2Cl_t^{\rm edg}=\kappa_{\rm e}\left(f_t^{\rm s}\right)^3T,
\end{equation}
respectively, where $\kappa_{\rm e}$ is the energy efficiency for edge computing. Then, the total energy consumption of the WD in time slot $t$ on data sensing and processing is
\begin{equation}
e_t^{\rm u} = e_t^{\rm col} + e_t^{\rm off} + e_t^{\rm loc}.
\end{equation}

\subsection{Task Data Queue Model}
The data sensed in the $t$th time slot is ready for processing at the beginning of the $(t+1)$th time slot. Let $Q_t^{\rm U}$ and $Q_t^{\rm S}$ be the data queue length in the WD and MS at the start of slot $t$, respectively. For analytical tractability, we assume infinite task queue capacity. For both the data queues at the MS and WD, the data processed within the current time slot cannot exceed the data queue backlog, i.e., 
\begin{equation}\label{ledge_off_cons}
0\leq l_t^{\rm off}+l_t^{\rm loc} \leq Q_{t}^{\rm U}, ~~ 0\leq f_t^{\rm s}T/C\leq Q_{t}^{\rm S}, \forall t=0, 1, \cdots.
\end{equation}
As a result, the data queues of the WD and MS evolve as following
\begin{equation}\label{Data_queue_simp}
Q_{t+1}^{\rm U} = Q_t^{\rm U}-l_t^{\rm off}-l_t^{\rm loc} + r_t,~Q_{t+1}^{\rm S} = Q_{t}^{\rm S} - l_t^{\rm edg} + l_t^{\rm off},~t=1,2,\cdots.
\end{equation}
To maintain stable data queues at the WD and MS, we consider stability constraints on $Q_t^{\rm U}$ and $Q_t^{\rm S}$ as following \cite{Neely2010}
\begin{equation}\label{DataQ_stab}
\bar{Q}_{\rm U} = \lim_{N\to+\infty}\frac{1}{N}\sum_{t=1}^{N}\mathbb{E}\left[Q_t^{\rm U}\right]<\infty,~\bar{Q}_{\rm S} = \lim_{N\to+\infty}\frac{1}{N}\sum_{t=1}^{N}\mathbb{E}\left[Q_t^{\rm S}\right]<\infty,
\end{equation}
where the expectation is taken over all random events, i.e., random channel fading, energy arrivals and operation state of the primary link. 
\subsection{Energy Queue Model}	
We model the EH process at the WD as a stochastic process with random energy arrivals in different time slots. In particular, $e_t^{\rm h}$ Joules of energy arrives at WD throughout the $t$-th time slot, where $e_t^{\rm h}$ is bounded above by $E_{\rm max}^{\rm h}$. Let $B_t$ be the battery level of the WD at the beginning of time slot $t$. We consider an energy-aware battery management policy: when $B_t$ is lower than a threshold $B_{\rm min}$, the WD stops consuming energy on data sensing and processing, while only harvesting ambient energy to replenish the battery. As a result, the energy consumed by the WD in the $t$th slot must satisfy
\begin{equation}\label{EH_caus}
0\leq \lambda_{\rm e}e_t^{\rm u}\leq B_t\cdot\mathbbm{1}_{B_t\geq B_{\rm min}},
\end{equation}
where $\lambda_{\rm e}$ is a non-negative scaling factor (e.g., $\lambda_e = 1000$ denotes using mJ as the unit) and $\mathbbm{1}_{\{\cdot\}}$ is the indicator function. The dynamics of battery level is:
\begin{equation}\label{E_evol}
B_{t+1} = \min\left(B_t - \lambda_{\rm e}e_t^{\rm u} + \lambda_{\rm e}e_t^{\rm h}, \Omega\right),
\end{equation}
where $\Omega$ is the battery capacity. 

\subsection{Problem Formulation}
In this paper, we aim at maximizing the long-term average data sensing rate of the WD under system stability and interference constraints. Given limited energy harvested from the ambient environment, this requires the WD to judiciously optimize the data sensing and task computation operations in sequential time slots. Denote the objective function as $\bar{R}=\lim_{N\to+\infty}\frac{1}{N}\sum_{t=0}^{N-1}r_t$. We formulate the target problem as below:
\begin{subequations}\label{Primal_prob}
	\begin{align}
	(\text{P1})~~\underset{\substack{r_t, p_t^{\rm u}, f_t^{\rm u}, f_t^{\rm s},\forall t}} \max~&\bar{R}\\
	\st
	~~& \eqref{SNR_thre}, \eqref{ledge_off_cons}, \eqref{DataQ_stab}, \eqref{EH_caus}\\
	~~& \lim_{N\to+\infty}\frac{1}{N}\sum_{t=1}^{N}\mathbb{E}\left[e_t^{\rm edg}\right] \leq c_{\rm th}, \label{Bud_cons}\\
	~~& 0 \!\leq r_t \!\leq\! r_{\rm max}, 0 \!\leq p_t \!\leq\! p_{\rm max},\! 0 \leq\! f_t^{\rm u}\! \leq\! f^{\rm u}_{\rm max}, \!0 \leq\! f_t^{\rm s}\! \leq\! f^{\rm s}_{\rm max}, \forall t, \label{Prob_const2}
	\end{align}
\end{subequations}
where \eqref{SNR_thre} is the QoS constraint of the primary link. \eqref{ledge_off_cons} and \eqref{EH_caus} are the data causality and energy causality at the WD, repetitively. \eqref{DataQ_stab} denotes the data queue stability constraints. \eqref{Bud_cons} denotes the average power constraint at the MS, where $c_{\rm th}$ is the power threshold. At the beginning of time slot $t$, we assume perfect knowledge of the current system state $s_t=\{\omega_t, I_t\}$ at the MS, where $\omega_t=\{a_t,e_t^{h},h_t,\bar{g}_t,\bar{h}_t\}$ captures the environment random events and $I_t=\{Q_t^{\rm S},Q_t^{\rm U},B_t\}$ is the queue backlog state. We seek an online algorithm that makes control decisions $\{r_t, p_t^{\rm u}, f_t^{\rm u}, f_t^{\rm s}\}$ in the $t$th slot based only on $s_t$. The difficulty of the online design is twofold. First, under the stochastic channels and random access of the primary user, it is hard to meet the long-term requirements when the decisions are made in each time slot without knowing the future system information. Second, due to the random energy supply at the WD, the system decisions in different slots are inherently coupled with each other. This poses great challenge to allocate the battery energy to strike a good balance between the current and future system performance. In the following, we propose a \textbf{P}erturbed-\textbf{Ly}apunov-based online data \textbf{S}ensing and \textbf{E}dge computation (PLySE) algorithm to solve (P1), which controls data sensing and processing in an online manner without requiring a priori knowledge of the system state.

\section{Online Data Sensing and Computation Offloading Optimization}\label{sec4}
\subsection{Perturbed Lyapunov-based Optimization}
Lyapunov optimization is a well-established method to design online algorithm with long-term stability requirement. However, standard Lyapunov optimization technique is not directly applicable to solve (P1) because the feasible control action sets are coupled over time due to the temporally correlated battery energy in constraint \eqref{EH_caus}. Here, we introduce a perturbed Lyapunov method to circumvent this issue. To start with, we introduce for the WD a perturbed battery queue, i.e., 
\begin{equation}\label{Queue_B}
\tilde{B}_t \triangleq  B_t - \Omega.
\end{equation} 
As shown in Section \ref{secV}, by employing a sufficient large battery capacity $\Omega$, we can safely remove the constraint \eqref{EH_caus} without violating the energy causality and at the same time decouple the feasible control action sets in different time slots. We define for the MS a virtual power deficit queue with the update equation:
\begin{equation}\label{Queue_Z}
Z_{t+1} =  \max\left(Z_t + \lambda_{\rm c} e_t^{\rm edg} - \lambda_{\rm c} c_{\rm th},0\right), t=1,2,\cdots,
\end{equation}
where $\lambda_{\rm c}$ is a positive scaling factor. Intuitively, the average power consumption constraint in \eqref{Bud_cons} is satisfied if $Z_t$ is finite as $t\rightarrow \infty$. We rewrite the system queue backlog as $\Theta_t\triangleq\left\{\tilde{B}_t,Z_t,Q_t^{\rm U},Q_t^{\rm S}\right\}$. Further, we define the perturbed Lyapunov function as
\begin{equation}
\Phi_t = \frac{1}{2}\left(\tilde{B}_t\right)^2 + \frac{1}{2}\left(Z_t\right)^2 + \frac{1}{2}\left(Q_t^{\rm U}\right)^2 + \frac{1}{2}\left(Q_t^{\rm S}\right)^2,
\end{equation}
and the Lyapunov drift as
\begin{equation}\label{LyaFunc}
\Delta^t = \mathbb{E}\left[\Phi_{t+1}-\Phi_t|\Theta_t\right],
\end{equation}
where the expectation is with respect to the system random processes given the current system queue state $\Theta_t$. To maximize the data sensing rate $\bar{R}$ while maintaining stable system queue $\Theta_t$, we adopt the drift-plus-penalty approach \cite{Neely2010}, which greedily minimizes an upper bound of the following Lyapunov drift-plus-penalty function in each time slot: 
\begin{equation}\label{LyaDrift}
\Delta^t_V=\Delta^t-V\mathbb{E}\left[r_t|\Theta_t\right].
\end{equation}
where $V$ is a positive weight factor. 

In the following, we provide an upper bound of $\Delta^t_V$. For convenience, we denote a constant
\begin{equation}\label{Dvalue}
D\!\!=\!\!\frac{1}{2}\!\!\left[\left(\lambda_{\rm c}e_{\rm max}^{\rm edg}\right)^2\!\!+\!\!\left(\lambda_{\rm c}c_{\rm th}\right)^2\!\!+\!\!\left(\lambda_{\rm e}e_{\rm max}^{\rm u}\right)^2\!\!+\!\!\left(\lambda_{\rm e}E_{\rm max}^{\rm h}\right)^2\!\!+\!\left(l_{\rm max}^{\rm off}\!\!+\!\!l_{\rm max}^{\rm loc}\right)^2\!\!+\!\!r_{\rm max}^2+\!\left(l_{\rm max}^{\rm edg}\right)^2\!\!+\!\left(l_{\rm max}^{\rm off}\right)^2\!\right],
\end{equation}
where $e_{\rm max}^{\rm edg}\triangleq\kappa_{\rm c}\left(f_{\rm max}^{\rm s}\right)^3T$ is the largest per-slot energy cost at the MS. $e_{\rm max}^{\rm u}\!\triangleq\!e_{\rm unit}^{\rm col}r_{\rm max}\!+\!p_{\rm max}T\!+\!\kappa_{\rm e}\left(f^{\rm u}_{\rm max}\right)^3T$ is the maximum per-slot energy cost at the WD. $l_{\rm max}^{\rm off}\!\triangleq\! \mathbb{E}\left[WT\log_2\left(1\!+\!p_{\rm max}\gamma_t\right)\right]$ corresponds to the maximum average transmission rate of the WD. $l_{\rm max}^{\rm edg}\triangleq\frac{f^{\rm s}_{\rm max}T}{C}$ and $l_{\rm max}^{\rm loc}\triangleq\frac{f^{\rm u}_{\rm max}T}{C}$ are the maximum data processing rate for edge computing and local computing, respectively. 
\begin{Lemma}\label{lem6}
	Under any control method, the Lyapunov drift-plus-penalty function \eqref{LyaDrift} has the following upper bound for all $t$, all possible values of $\Theta_t$, and all parameters $V\geq 0$:
	\begin{equation} \label{LyaDriPlusPen}
	\begin{split}
	\Delta^t_V &\leq D-\! \mathbb{E}\left\{Vr_t \!+ \!Z_t\lambda_{\rm c}\left(c_{\rm th}\!-\!e_t^{\rm edg}\right)\!+\!\lambda_{\rm e}(B_t\!-\!\Omega)(e_t^{\rm u}\!-\!e_t^{\rm h})\right.\\
	&\left. +Q_t^{\rm U}\left(l_t^{\rm off}+l_t^{\rm loc}-r_t\right)+\!Q_t^{\rm S}(l_t^{\rm edg}-l_t^{\rm off})\mid\!\!\Theta_t\right\}.
	\end{split}
	\end{equation}
\end{Lemma}
\begin{proof}
	Please refer to Appendix \ref{app0} for detail.
\end{proof}

\renewcommand{\algorithmicrequire}{\textbf{Initialization:}}
\renewcommand{\algorithmicensure}{\textbf{Output:}}
\begin{algorithm}
	\caption{The online PLySE algorithm to solve \eqref{Primal_prob}}
	\label{alg1}
	\begin{algorithmic}[1]
		\REQUIRE The initial system state $s_0=\{\omega_0,\Theta_0\}$, where $\Theta_0 = \{Q_{0}^{\rm S},Q_{0}^{\rm U},\tilde{B}_{0},Z_{0}\}$.
		\FOR {each time slot $t$}
		\STATE Observe the system state $s_t$.
		\STATE Solve problem \eqref{Prob_Per_slot} for $\{r_t^\ast, f_t^{\rm s\ast}, f_t^{\rm u\ast}, p_t^{\rm u\ast}\}$ using \eqref{Opt_fs}, \eqref{Opt_r} and \eqref{Opt_pu_fu}.
		\STATE Execute the control action $\{r_t^\ast, f_t^{\rm s\ast}, f_t^{\rm u\ast}, p_t^{\rm u\ast}\}$ and update $\Theta_{t+1}=\left\{Q_{t+1}^{\rm S},Q_{t+1}^{\rm U},\tilde{B}_{t+1},Z_{t+1}\right\}$ according to \eqref{Data_queue_simp}, \eqref{Queue_B} and \eqref{Queue_Z}, respectively.
		\ENDFOR
	\end{algorithmic}
\end{algorithm}

With Lemma 1, we illustrate the proposed PLySE method to solve (P1) in Algorithm 1. In time slot $t$, PLySE observes the current system states $s_t$ and minimizes the right hand side of \eqref{LyaDriPlusPen}. Specifically, it determines the actions of the MS and WD in time slot $t$ by solving the following problem:
\begin{subequations}\label{Prob_Per_slot}
\begin{align}
\underset{\substack{r_t,p_t^{\rm u}\\ f_t^{\rm u},f_t^{\rm s}}} \max~&Vr_t \!+\! \!Z_t\lambda_{\rm c}\left(c_{\rm th}\!-\!e_t^{\rm edg}\right)\!+\!\lambda_{\rm e}(B_t\!-\!\Omega)(e_t^{\rm u}\!-\!e_t^{\rm h})+\!Q_t^{\rm U}\left(l_t^{\rm off}+l_t^{\rm loc}\!-\!r_t\right)+\!Q_t^{\rm S}(l_t^{\rm edg}\!-\!l_t^{\rm off})\label{Prob_Per_slot_obj}\\
\st
~~& a_t\left(p_{t}^{\rm u}\bar{h}_t+W\delta_{\rm p}^2-\Gamma_{\rm th}\right)\leq 0, \label{Prob_Per_slot_cons1}\\
~~& l_t^{\rm off}+l_t^{\rm loc} \leq Q_t^{\rm U},\\
~~& 0\leq f_t^{\rm s}T/C\leq Q_t^{\rm S}, \label{Prob_Per_slot_cons3}\\
~~& 0 \!\leq \!r_{t} \!\leq \!r_{\rm max}, 0 \!\leq \!p_{t}^{\rm u} \!\leq \!p_{\rm max}, \!0 \leq\! f_t^{\rm u} \!\leq\! f_{\rm max}^{\rm u}, 0 \!\leq \!f_t^{\rm s}\! \leq \!f_{\rm max}^{\rm s}.
\end{align}
\end{subequations}
Comparing with (P1), we remove the energy causality constraint \eqref{EH_caus} in the per-slot problem \eqref{Prob_Per_slot}. In Section \ref{secV}, we show that \eqref{EH_caus} can always be satisfied when implementing PLySE given that the battery capacity satisfies a mild condition. A close observation shows that \eqref{Prob_Per_slot} can be decomposed into three independent subproblems, which correspond to CPU frequency control at the MS, data sensing control at the WD, and task execution control at the WD, respectively. These three subproblems can be solved in parallel as follows.
\subsection{Optimal Edge CPU Frequency}
The optimal CPU frequency at the MS can be obtained by solving the following convex optimization problem:
\begin{subequations}\label{Subprob_fs}
	\begin{align}
	\underset{\substack{f_t^{\rm s}}} \max~& - Z_t\lambda_{\rm c} \kappa_{\rm e}\left(f_t^{\rm s}\right)^3T + Q_t^{\rm S}f_t^{\rm s}T/C \label{Subprob_fs_obj}\\
	\st
	~~& 0 \leq f_t^{\rm s} \leq \bar{f}_{\rm max}^{\rm s},
	\end{align}
\end{subequations}
where $\bar{f}_{\rm max}^{\rm s}=\min(Q_t^{\rm S}C/T, f_{\rm max}^{\rm s})$, which is obtained by absorbing \eqref{Prob_Per_slot_cons3} into the box constraint $0 \!\leq \!f_t^{\rm s}\! \leq \!f_{\rm max}^{\rm s}$. The solution of \eqref{Subprob_fs} can be easily obtained as
\begin{equation}\label{Opt_fs}
f_t^{\rm s\ast}=\min\left(\sqrt{\frac{Q_t^{\rm S}}{3Z_t\lambda_{\rm c} C\kappa_{\rm e}}},\bar{f}_{\rm max}^{\rm s}\right).
\end{equation}
As shown in \eqref{Opt_fs}, the MS operates at a high CPU frequency when the data queue length $Q_t^{\rm S}$ is large, and slows down when the power deficit queue $Z_t$ is large. Such an operation stabilizes the data queue $Q_t^{\rm S}$ and satisfies the long-term energy budget at the MS.

\subsection{Optimal Data Sensing Rate}
The optimal task data size collected in time slot $t$ can be obtained by solving the following linear programming:
\begin{equation}
\underset{\substack{0\leq r_t\leq r_{\rm max}}}\max\left[V+\lambda_{\rm e}(B_t\!-\!\Omega)e_{\rm unit}^{\rm col}-Q_t^{\rm U}\right]r_t,
\end{equation}
where the optimal solution exhibits a simple ON-OFF structure:
\begin{equation}
r^{\ast}_t = r_{\rm max}\cdot \mathbbm{1}_{C_{\rm sen}\leq 0}, \forall t, \label{Opt_r}
\end{equation}
where $C_{\rm sen}\triangleq Q_t^{\rm U}-V-\lambda_{\rm e}(B_t\!-\!\Omega)e_{\rm unit}^{\rm col}$. Specifically, the WD senses task data at the maximum rate (i.e., $r_t=r_{\rm max}$) if $C_{\rm sen}\leq 0$, and collects zero-bit data otherwise. Because $C_{\rm sen}$ increases with $Q_t^{\rm U}$ and decreases with $B_t$, the WD reduces sensing activity when $Q_t^{\rm U}$ is large or $B_t$ is small, thus avoiding continuous data queue backlog and energy draining at the WD.

\subsection{Optimal Task Execution}
The remaining sub-problem optimizes the task execution, including local computing and task offloading control at the WD. After removing the terms that are only related to $f_t^{\rm s}$ and $r_t$ in \eqref{Prob_Per_slot}, we solve the following optimization problem:
\begin{subequations}\label{Prob_Per_slot_sim}
	\begin{align}
	\underset{\substack{p_t^{\rm u}, f_t^{\rm u}}} \max~&F(f_t^{\rm u}) + G(p_t^{\rm u}) \label{Prob_Per_slot_sim_redc_obj}\\
	\st
	~~&l_t^{\rm off}+l_t^{\rm loc} \leq Q_t^{\rm U}, \label{Prob_Per_slot_sim_cons_data}\\
	~~& 0 \!\leq \!p_{t}^{\rm u} \!\leq \!p_{\rm th}, ~~0 \leq\! f_t^{\rm u} \!\leq\! f_{\rm max}^{\rm u}. \label{Prob_Per_slot_sim_cons_pf}
	\end{align}
\end{subequations}
Here, $p_{\rm th}=a_t\min\left(\frac{\Gamma_{\rm th}-W\delta_{\rm p}^2}{\bar{h}_t}, p_{\rm max}\right) +(1-a_t)p_{\rm max}$, which is obtained by absorbing $p_{t}^{\rm u}\leq p_{\rm max}$ into \eqref{Prob_Per_slot_cons1}. Besides, the objective of \eqref{Prob_Per_slot_sim} is detailed as,
\begin{equation}\label{Eq_F}
F(f_t^{\rm u})=\lambda_{\rm e}\tilde{B}_t\kappa_{\rm c}\left(f_t^{\rm u}\right)^3T\!+\!Q_t^{\rm U}\frac{f_t^{\rm u}T}{C},
\end{equation}
\begin{equation}\label{Eq_G}
G(p_t^{\rm u})=\lambda_{\rm e}\tilde{B}_tp_{t}^{\rm u}T\!+\!\left(Q_t^{\rm U}-Q_t^{\rm S}\right)TW\log_2\left(1+p_{t}^{\rm u}\gamma_t\right).
\end{equation}

\begin{Remark}
	In \eqref{Prob_Per_slot_sim}, the proposed PLySE method optimizes a weighted summation of energy cost and data processing rate (see \eqref{Eq_F} and \eqref{Eq_G}) at the WD. At a low battery level, the weighting factor of energy cost (i.e., $\tilde{B}_t$) has a large absolute value $|\tilde{B}_t|$. In this case, the WD prefers energy conservation to data processing (i.e., task offloading and local computation). On the contrary, when the battery level is high (i.e., $|\tilde{B}_t|$ is small), the WD tends to utilize the harvested energy for data processing. This yields a closed loop control on battery level which improves the energy efficiency at the WD. 
\end{Remark}

Let $\mathcal{F}_p(x) \triangleq \frac{1}{\gamma_t}2^{\frac{1}{WT}\left(Q_t^{\rm U}-\frac{xT}{C}\right)}-\frac{1}{\gamma_t}$ and $\mathcal{F}_f(x) \triangleq \frac{\left[Q_t^{\rm U}-WT\log_2\left(1+x\gamma_t\right)\right]C}{T}$. From \eqref{Prob_Per_slot_sim_cons_data} and \eqref{Prob_Per_slot_sim_cons_pf}, we can equivalently express the feasible region of \eqref{Prob_Per_slot_sim} as $p_t^{\rm u}\in[0, \bar{p}_{\rm th}]$ and $f_t^{\rm u}\in[0, \bar{f}_{\rm th}]$, where $\bar{p}_{\rm th}=\min\left(p_{\rm th}, \mathcal{F}_p(0)\right)$, and $\bar{f}_{\rm th}=\min\left(f_{\rm max}^{\rm u}, \mathcal{F}_f(0)\right)$. Due to constraint \eqref{Prob_Per_slot_sim_cons_data} and time-varying coefficient $\left(Q_t^{\rm U}-Q_t^{\rm S}\right)$ in $G(p_t^{\rm u})$, \eqref{Prob_Per_slot_sim} is generally a non-convex optimization problem. In the following Lemma \ref{lem_opt_solution}, we derive the closed-form expression of optimal solution of \eqref{Prob_Per_slot_sim}.

\begin{prop}\label{lem_opt_solution}
	The optimal solution of \eqref{Prob_Per_slot_sim} is
	\begin{subnumcases}{\left\{f_t^{\rm u\ast},p_t^{\rm u\ast}\right\} =\label{Opt_pu_fu}}
	\left\{\hat{f}_t^{\rm u}, \hat{p}_t^{\rm u}\right\}, \!& \text{if $Q_t^{\rm U}-Q_t^{\rm S}\geq0$ and $\hat{l}_{\rm max}^{\rm off}+\hat{l}_{\rm max}^{\rm loc} \leq Q_t^{\rm U}$}, \label{Opt_pu_fu_3}\\
	\left\{\hat{f}_t^{\rm u}, \mathcal{F}_p\left(\hat{f}_t^{\rm u}\right)\right\}, \!& \text{if $Q_t^{\rm U}-Q_t^{\rm S}\geq0$ and $\hat{l}_{\rm max}^{\rm off}+\hat{l}_{\rm max}^{\rm loc} > Q_t^{\rm U}$ and $\tilde{B}_t=0$}, \label{Opt_pu_fu_2}\\
	\left\{\breve{f}_t^{\rm u}, \mathcal{F}_p\left(\breve{f}_t^{\rm u}\right)\right\}, \!& \text{if $Q_t^{\rm U}-Q_t^{\rm S}\geq0$ and $\hat{l}_{\rm max}^{\rm off}+\hat{l}_{\rm max}^{\rm loc} > Q_t^{\rm U}$ and $\tilde{B}_t<0$}, \label{Opt_pu_fu_4}\\
	\left\{\hat{f}_t^{\rm u}, 0\right\}, \!& \text{if $Q_t^{\rm U}-Q_t^{\rm S}<0$}.\label{Opt_pu_fu_5}
	\end{subnumcases}
	Here, $\hat{f}_t^{\rm u}\!=\!\min\left(\sqrt{\frac{-Q_t^{\rm u}}{3\lambda_{\rm e}\tilde{B}_t\kappa_{\rm c}C}}, \bar{f}_{\rm th}\right)$ and $\hat{p}_t^{\rm u}\!\!=\!\!\left[\frac{\left(Q_t^{\rm S}-Q_t^{\rm U}\right)W}{\lambda_{\rm e}\tilde{B}_t\ln2}\!-\!\frac{1}{\gamma_t}\right]_0^{\bar{p}_{\rm th}}$, with $[\cdot]_x^y=\min(\max(\cdot,x),y)$. $\breve{f}_t^{\rm u}\!=\!\left[\bar{f}_t^{\rm u}\right]_{f_{\rm I}^{\rm lb}}^{f_{\rm I}^{\rm ub}}$, with $f_{\rm I}^{\rm ub} = \hat{f}_t^{\rm u}$, $f_{\rm I}^{\rm lb} = \max\left(0, \mathcal{F}_f\left(\hat{p}_t^{\rm u}\right)\right)$, and $\bar{f}_t^{\rm u}\in[0,+\infty)$ is the unique solution of 
	\begin{equation}\label{EqLemCase3}
	U^\prime(f_t^{\rm u}) = 3\lambda_{\rm e}\kappa_{\rm c}T\tilde{B}_t\left(f_t^{\rm u}\right)^2 - \frac{\lambda_{\rm e}\tilde{B}_tT\ln 2}{WC\gamma_t}2^{\frac{Q_t^{\rm U}}{WT}-\frac{f_t^{\rm u}}{WC}} + \frac{T}{C}Q_t^{\rm S} = 0.
	\end{equation}
	In particular, $U^\prime(f_t^{\rm u})$ is a monotonically decreasing function of $f_t^{\rm u}$, and thus $\bar{f}_t^{\rm u}$ can be obtained via bisection search.
\end{prop}
\begin{proof}
	Please refer to Appendix \ref{App_Opt_solution} for detail.
\end{proof}

\begin{Remark}	
	From Proposition \ref{lem_opt_solution}, we see that the optimal edge computation control solutions are directly affected by the current available energy $B_t$ (absorbed in $\tilde{B}_t$), time-varying data queue length $\{Q_t^{\rm U}, Q_t^{\rm S}\}$ and primary link activity $a_t$, detailed as following: a) The local CPU frequency $\hat{f}_t^{\rm u}$ and transmit power $\hat{p}_t^{\rm u}$ increase with the current available energy $B_t$. b) A larger data queue $Q_t^{\rm U}$ yields a higher local CPU frequency $\hat{f}_t^{\rm u}$. c) The optimal offloading solution follows a threshold-based structure: the WD offloads to the MS only when the local data queue length is longer than that in the edge, i.e., $Q_t^{\rm U}-Q_t^{\rm S}\geq 0$, otherwise it only performs local computation. Besides, the larger the difference of $Q_t^{\rm U}-Q_t^{\rm S}$, the higher the transmit power $\hat{p}_t^{\rm u}$ at the WD. d) If the primary link is active in the $t$th slot, the maximum allowable transmit power at the WD $p_{\rm th}$ monotonically decreases with the interference threshold $\Gamma_{\rm th}$. e) A larger active probability $\bar{a}$ imposes a stringent transmit power constraint to $p_{\rm th}$ in more time slots, which eventually reduces the task offloading rate in the long-term. In contrast, since the CPU frequency lower bound $f_{\rm I}^{\rm lb} = \max\left(0, \mathcal{F}_f\left(\hat{p}_t^{\rm u}\right)\right)$ decreases with $\hat{p}_t^{\rm u}$, a larger $\bar{a}$ would yield a higher local CPU frequency $\breve{f}_t^{\rm u}$. Overall, the PLySE algorithm tends to stabilize both $Q_t^{\rm U}$ and $Q_t^{\rm S}$, and satisfy the QoS requirements of primary link.
\end{Remark}

\section{Performance Analysis}\label{secV}
In this section, we analyze the performance of the proposed PlySE algorithm. Recall that we have removed the energy causality constraint \eqref{EH_caus} in problem (P1) when designing the PLySE algorithm. Here, we first show that \eqref{EH_caus} is always satisfied when implementing the online PLySE algorithm given that the battery capacity satisfies a mild condition. Then, we prove that the PLySE algorithm also satisfies all the long-term performance constraints, thus producing a feasible solution to (P1), and achieves an $\left[O(1/V),O(V)\right]$ sensing-delay tradeoff by tuning the Lyapunov parameter $V$. 

To start with, we derive in the following Lemma \ref{lem_ub} an upper bound for data queue $Q_t^{\rm U}$, which is useful to determine the threshold of battery capacity $\Omega$.
\begin{Lemma}\label{lem_ub}
	For a non-negative parameter $V$ and an initial data queue satisfying $Q_0^{\rm U}\in\left[0,Q_{\rm max}\right]$, where $Q_{\rm max} = V+r_{\rm max}$, we have $0\leq Q_t^{\rm U}\leq Q_{\rm max}$, for $t=0,1,2,\cdots$
\end{Lemma}
\begin{proof}
	We prove this result by induction. Notice that $0\leq Q_0^{\rm U}\leq Q_{\rm max}$ holds initially. In the following, we assume that $0\leq Q_t^{\rm U}\leq Q_{\rm max}$ holds in time slot $t$, and prove that $0\leq Q_{t+1}^{\rm U}\leq Q_{\rm max}$ by considering two cases:
	\begin{itemize}
		\item If the WD does not collect any data in time slot $t$, we obviously have that $Q_{t+1}^{\rm U}\leq Q_{t}^{\rm U} \leq V+r_{\rm max}$;
		\item If the WD collects $r_t$-bit data in time slot $t$, then $r_t=r_{\rm max}$ and $V+\lambda_{\rm e}(B_t\!-\!\Omega)e_{\rm unit}^{\rm col}-Q_t^{\rm U}\geq 0$ according to \eqref{Opt_r}. Therefore, we have that $Q_t^{\rm U} \leq V+\lambda_{\rm e}(B_t\!-\!\Omega)e_{\rm unit}^{\rm col}\leq V$. As a result, we can obtain that $Q_{t+1}^{\rm U}\leq Q_{t}^{\rm U} + r_{\rm max} \leq V+r_{\rm max}$.
	\end{itemize}
	This completes the proof that $0\leq Q_{t+1}^{\rm U}\leq Q_{\rm max}$.
\end{proof}

To facilitate our exposition, we denote $A_1 = \frac{3C\kappa_{\rm c}B_{\rm min}}{\kappa_{\rm e}\left(V+r_{\rm max}\right)T}$, $A_2 = -\frac{3CW\kappa_{\rm c}}{\kappa_{\rm e}\ln 2}$, $A_3 = -\frac{\left(V+r_{\rm max}\right)}{3\lambda_{\rm e}C\kappa_{\rm c}}$, $\bar{A}_1 = -\frac{A_2^2}{3A_1^2}$ and $\bar{A}_2 = \frac{-2A_2^3+27A_1A_3}{27A_1^3}$. The following Proposition \ref{lem3} provides a sufficient condition to remove the energy causality constraint \eqref{EH_caus} when implementing PLySE to solve (P1).

\begin{prop}\label{lem3}
	Suppose that $\Omega\!\geq\!\max\left(\!\frac{V}{\lambda_{\rm e}e_{\rm unit}^{\rm col}}\!+\!\lambda_{\rm e}e_{\rm max}^{\rm u},x_{\rm max}\!+\!\lambda_{\rm e}e_{\rm max}^{\rm u}\!\right)+\lambda_{\rm e}E_{\rm max}^{\rm h}$, where $x_{\rm max}=\max_{k}\left(x_k\right)$ and $x_k = 2\sqrt{-\frac{\bar{A}_1}{3}}\cos\left[\frac{1}{3}\arccos\left(\frac{3\bar{A}_2}{2\bar{A}_1}\sqrt{-\frac{3}{\bar{A}_1}}\right)-\frac{2\pi k}{3}\right], ~\text{for}~k=0,1,2$, the energy causality constraint \eqref{EH_caus} is satisfied in every time slot.
\end{prop}
\begin{proof}
	Please refer to Appendix \ref{app1} for detail.
\end{proof}
Proposition \ref{lem3} shows that the energy causality constraint can be safely removed when implementing the PLySE algorithm to solve (P1), as long as the battery capacity is sufficiently large. Such a condition easily holds in practice. For instance, the condition is $\Omega\geq 137.8$ Joules using the parameters in Section IV, which holds for commercial battery with several thousand Joules capacity. Moreover, Proposition \ref{lem3} shows that the required battery capacity increases with $V$. Together with the fact that a larger $V$ yields a higher data sensing rate (to be shown in Proposition \ref{lem4}), we observe a critical tradeoff between the battery capacity and achievable data sensing rate. That is, the data sensing rate improves with the battery capacity, and can achieve arbitrarily close-to-optimal sensing performance when the battery capacity is sufficiently large.

In the following, we show that the PLySE algorithm can achieve $\left[O(1/V),O(V)\right]$ sensing-delay tradeoff while respecting all the long-term performance constraints. To facilitate the performance analysis, we introduce the following auxiliary problem:
\begin{subequations}\label{Primal_prob_mod}
	\begin{align}
	(\text{P2})~~	\underset{\substack{r_t, p_t^{\rm u}, f_t^{\rm u}, f_t^{\rm s},\forall t}} \max~&\lim_{N\to+\infty}\frac{1}{N}\sum_{t=0}^{N-1}r_t\\
	\st
	~~& \eqref{SNR_thre}, \eqref{ledge_off_cons}, \eqref{DataQ_stab}, \eqref{Bud_cons}, \eqref{Prob_const2},\\
	~~& \lim_{N\to+\infty}\frac{1}{N}\sum_{t=0}^{N-1}\mathbb{E}\left[e_t^{\rm u}-e_t^{\rm h}\right] \leq 0. \label{Energy_lonterm}
	\end{align}
\end{subequations}
Compared to (P1), (P2) replaces the energy causality constraint (11) in (P1) with a long-term energy constraint \eqref{Energy_lonterm}. Let $\bar{R}_{\rm P1}^\ast$ and $\bar{R}_{\rm P2}^\ast$ be the optimal value of (P1) and (P2), respectively. Then, we show in the following lemma that $\bar{R}_{\rm P1}^\ast\leq \bar{R}_{\rm P2}^\ast$.
\begin{Lemma}\label{lem1}
	Any feasible solution to (P1) is also feasible to (P2). Specifically, (P2) is a relaxed version of (P1), i.e., $\bar{R}_{\rm P1}^\ast\leq \bar{R}_{\rm P2}^\ast$.
\end{Lemma}
\begin{proof}
	For any feasible solution of (P1), based on the battery dynamics \eqref{E_evol}, we have
	\begin{equation}
	B_{t+1} \leq B_t - \lambda_{\rm e}e_t^{\rm u} + \lambda_{\rm e}e_t^{\rm h}, t = 0, \cdots, N-1.
	\end{equation}
	By summing up both sides of the above $N$ equalities, taking the expectation, diving both sides by $\lambda_{\rm e}N$ and letting $N$ go to infinity, we have
	\begin{equation}
	\lim_{N\!\to\!+\!\infty}\!\frac{1}{\lambda_{\rm e}N}\mathbb{E}\left[B_{N}\right]\!\leq\! \lim_{N\!\to\!+\!\infty}\!\!\frac{1}{\lambda_{\rm e}N}\mathbb{E}\left[B_{0}\right]\! -\! \lim_{N\!\to\!+\!\infty}\!\frac{1}{N}\!\sum_{t=0}^{N-1}\!\mathbb{E}\left[e_t^{\rm u}-e_t^{\rm h}\right].
	\end{equation}
	Since $B_t\leq\Omega<+\infty$, we have $\lim_{N\to+\infty}\frac{1}{\lambda_{\rm e}N}\mathbb{E}\left[B_{N}\right]=\lim_{\lambda_{\rm e}N\to+\infty}\frac{1}{N}\mathbb{E}\left[B_{0}\right]=0$, i.e., \eqref{Energy_lonterm} is satisfied. Hence, any feasible solution of (P1) is also feasible to (P2), and thus $\bar{R}_{\rm P1}^\ast\leq \bar{R}_{\rm P2}^\ast$.
\end{proof}
We denote the environment random event $\omega_t$ of the considered problem as an i.i.d process. We introduce a class of stationary and randomized policies called $\omega$-only policy, which observes $\omega_t$ for each time slot $t$ and makes control decisions independent of the queue backlogs $\Theta_{t}$. To ensure the long-term requirements \eqref{DataQ_stab}, \eqref{Bud_cons} and \eqref{Energy_lonterm}, we assume (P2) is feasible and following Slater condition holds.
\begin{Assumption}
	There are constant $\epsilon>0$ and $\varphi(\epsilon)\leq \bar{R}_{\rm P2}^\ast$ and an $\omega$-only policy $\Gamma$ satisfying that
	\begin{subequations}\label{SLT}
		\setlength{\abovedisplayskip}{-5pt}
		\setlength{\belowdisplayskip}{-5pt}
		\begin{align}
		&\mathbb{E}\left[r_t^{\Gamma}\right] = \varphi(\epsilon),~
		\mathbb{E}\left[e_t^{\rm edg,\Gamma}\right] \leq  c_{\rm th}-\epsilon,~
		\mathbb{E}\left[e_t^{\rm u, \Gamma}-e_t^{\rm h,\Gamma}\right] \leq -\epsilon,\\
		&\mathbb{E}\left[l_t^{\rm{off},\Gamma}\right] \leq\mathbb{E}\left[l_t^{\rm{edg},\Gamma}\right]-\epsilon,~
		\mathbb{E}\left[r_t^{\Gamma}\right] \leq \mathbb{E}\left[l_t^{\rm{off},\Gamma}+l_t^{\rm{loc},\Gamma}\right] -\epsilon.
		\end{align}
	\end{subequations}
\end{Assumption}

In the following Proposition \ref{lem4}, we show that PLySE achieves asymptotic optimality to the primary problem (P1), while satisfying the long-term constraints \eqref{DataQ_stab} and \eqref{Bud_cons}.
\begin{prop}\label{lem4}
	Under the proposed PLySE method, we have that:
	\begin{itemize}
		\item [a)] The achieved time average expected computation rate, denoted as $\bar{R}_{\Psi}$, satisfies that
		\begin{equation}\label{Asy_opt_R}
		\begin{split}
		\bar{R}_{\Psi}\geq \bar{R}_{\rm P1}^\ast - \frac{D}{V}.
		\end{split}
		\end{equation}
		\item[b)] The data queue stability \eqref{DataQ_stab} are guaranteed. In particular, the data queue length $Q_t^{\rm U}$ and $Q_t^{\rm S}$ satisfy that
		\begin{subequations}\label{lem5_eq2}
			\begin{align}
			&\lim_{N\!\to\!+\!\infty}\frac{1}{N}\sum_{t=0}^{N}\mathbb{E}\left[Q_t^{\rm U}\right] \leq \frac{D+V\left(\bar{R}_{\rm P1}^\ast-\varphi(\epsilon)\right)}{\epsilon}<+\infty,\\
			&\lim_{N\!\to\!+\!\infty}\frac{1}{N}\sum_{t=0}^{N}\mathbb{E}\left[Q_t^{\rm S}\right] \leq \frac{D+V\left(\bar{R}_{\rm P1}^\ast-\varphi(\epsilon)\right)}{\epsilon}<+\infty.
			\end{align}
		\end{subequations}
		\item[c)] $Z_t$ is strongly stable, and the long-term average power constraints \eqref{Bud_cons} is satisfied.
	\end{itemize}
\end{prop}
\begin{proof}
	Please refer to Appendix \ref{app2} for detail.
\end{proof}

According to Little's law, the network delay is proportional to the time-averaged data queue length. Proposition \ref{lem4} indicates that the PLySE algorithm achieves an $\left[O(1/V),O(V)\right]$ sensing rate-delay tradeoff. In particular, as $V$ increases, the sensing rate improves at the rate of $O(1/V)$, but at the cost of longer data queue length (processing delay) increasing at rate of $O(V)$.

\section{Numerical Results}\label{sec6}
In this section, we evaluate the system performance via numerical simulations. We consider an MS deployed at $d_g = 500$ meters away from the PT. The communication distance between the MS and the WD is $d_h=50$ meters and that between the WD and the PR is $d_{\bar{h}} = 50$ meters. We model all the channels as Rayleigh fading channels. Denote $\sigma_g$, $\sigma_h$, and $\sigma_{\bar{h}}$ as the path-loss exponents of channel MS-PT, MS-WD, and WD-PR, respectively. Then, we model the corresponding channel gain as $g_t=\varsigma_t{H}(d_g,\sigma_g)$, $h_t=\varsigma_t{H}(d_h,\sigma_h)$, and $\bar{h}_t=\varsigma_t{H}(d_{\bar{h}},\sigma_{\bar{h}})$, respectively. Here, $\varsigma_t$ is an independent exponential random variable of unit mean, which captures the small-scale channel fading effect in time slot $t$. ${H}(d,\sigma)$ denotes the average channel gain that follows a path-loss model ${H}(d,\sigma) = G_{\rm A}\left(\frac{3\times10^8}{4\pi f_{\rm c}d}\right)^\sigma$, where $G_{\rm A}=4.11$ captures the total antenna gain and $f_{\rm c}=2.4$ GHz represents the carrier frequency. Unless otherwise statement, we set $\sigma_{g}=\sigma_{h}=\sigma_{\bar{h}} = 2.7$, the power budget at the MS $c_{\rm th}=1.6$ Joules, and the maximum sensing data size $r_{\rm max}=10$ Mbits. We initialize the data queue length and battery level to 0, i.e., $Q_0^{\rm S}=Q_0^{\rm U}=B_0=0$. The energy arrival rate $E_t^{\rm h}$ is uniformly distributed in $\left[0,E_{\rm max}^{\rm h}\right]$ with $E_{\rm max}^{\rm h}=0.6$ Joules. We consider the operation state indicator $a_t$ of the PT follows a binomial distribution with expectation $\bar{a}=0.6$. We refer to $\bar{a}$ as the active rate of PT. Besides, we set the interference tolerance at the PR as $\Gamma_{\rm th}=5^3\times W\delta_{\rm p}^2$. The simulation length is set to $N=6\times10^4$ time slots. The other parameters used in simulation are listed in Table \ref{tab1}.

\begin{table}
	\centering
	\caption{Simulation Parameters}
	\label{tab1}
	\begin{tabular}{|c|c|c|c|}
		\hline
		$P_{\rm B} = 33$ dBm & $p_{\rm max}=20$ dBm& $e_{\rm unit}^{\rm col}=10^{-8}$ Joules/bit & $T = 1$ sec \\
		\hline
		$W = 1$ MHz & $k_{\rm c}=k_{\rm e}=10^{-26}$  & $C = 100$ cycles/bit &  $B_{\rm min} = 10^{-3}$ Joules\\
		\hline
		$\delta_{\rm s}^2 = \delta_{\rm p}^2 = -174$ dBm/Hz  & $f_{\rm s} = 4$ GHz & $f_{\rm max}^{\rm u} = 400$ MHz & $V=256\times10^7$ \\
		\hline	
	\end{tabular}
\end{table}

To verify the performance of the proposed PLySE method, we consider three representative methods as the benchmarks, all of which computes $f_t^{\rm s\ast}$ and $r_t^{\ast}$ similar to PLySE (i.e., using \eqref{Opt_fs} and \eqref{Opt_r}, respectively), while determining $p_t^{\rm u\ast}$ and $f_t^{\rm u\ast}$ via following strategies:
\begin{itemize}
	\item Local computing only (LCO): WD computes all tasks locally rather than task offloading. In this case, the optimal local CPU frequency $f_{t}^{\rm u\ast}$ is obtained by solving \eqref{Prob_Per_slot_sim} with $p_{t}^{\rm u\ast}=0$. Specifically, $f_{t}^{\rm u\ast} = \hat{f}_t^{\rm u}$.
	\item Edge computing only (ECO): WD offloads all tasks to the MS for edge processing. In this case, the optimal transmit power $p_{t}^{\rm u\ast}$ is obtained by solving \eqref{Prob_Per_slot_sim} with $f_{t}^{\rm u \ast}=0$. Specifically, $p_{t}^{\rm u\ast}=\hat{p}_t^{\rm u}$ if $Q_t^{\rm U}-Q_t^{\rm S}\geq0$ or $p_{t}^{\rm u\ast}=0$, otherwise.  
	\item $Q_{\rm S}$-oblivious computation offloading ($Q_{\rm S}$-oblivious): Instead of computing $p_t^{\rm u\ast}$ and $f_t^{\rm u\ast}$ using \eqref{Opt_pu_fu} where $Q_t^{\rm S}$ acts as a key factor,  $Q_{\rm S}$-oblivious method ignores the state of $Q_{\rm S}$ and obtains $p_t^{\rm u\ast}$ and $f_t^{\rm u\ast}$ by solving the following problem:
	\begin{equation}\label{Prob_Myopic}
		\underset{\substack{p_t^{\rm u}, f_t^{\rm u}}} \max~l_t^{\rm loc}+l_t^{\rm off},~~\st~~\eqref{EH_caus}, \eqref{Prob_Per_slot_sim_cons_data}, \eqref{Prob_Per_slot_sim_cons_pf}.
	\end{equation}
	Constrained by \eqref{Prob_Per_slot_sim_cons_data}, \eqref{Prob_Myopic} is a non-convex optimization problem. Nonetheless, we can obtain the optimal solution of \eqref{Prob_Myopic} following similar steps to solve \eqref{Prob_Per_slot_sim}, where the detail is omitted here for concision.
\end{itemize}

\subsection{Feasibility of PLySE and Benchmark Methods}
We first investigates the feasibility of the proposed PLySE algorithm and the three benchmark methods. For convenience, we denote the average energy consumption at the MS as $\bar{c} = \lim_{N\to+\infty}\frac{1}{N}\sum_{t=1}^{N}\mathbb{E}\left[e_t^{\rm edg}\right]$. We consider two different average power budgets at the MS, i.e., $c_{\rm th}=1.6$ and 0.2 Joules, and plot in Fig. \ref{L2_0} the average data queue length $\bar{Q}_{\rm U}$ and $\bar{Q}_{\rm S}$ as well as the average energy consumption at the MS (denoted as $\bar{c}$) as the time proceeds, where each point is a moving-window average of 400 time slots. We also display in the figure the change of battery level $B_t$ over time slot. The results in Fig. \ref{L2_0}(a)-(d) show that all the methods provide feasible solutions to (P1) at a high average power budget $c_{\rm th}=1.6$. In particular, all the methods stabilize the data queues $\bar{Q}_{\rm U}$ and $\bar{Q}_{\rm S}$ in Fig. \ref{L2_0}(a) and (b), respectively, where the proposed PLySE method achieves relative shorter data queues than ECO. Besides, they all satisfy the average energy consumption constraint 1.6 Joules in Fig. \ref{L2_0}(c). We also observe that the battery levels fluctuate between 0 and the battery capacity over time, which means that all the methods respect the energy causality constraint \eqref{EH_caus} in every time slot. However, when we decrease $c_{\rm th}$ from 1.6 to 0.2, $Q_{\rm S}$-oblivious yields an unstable $\bar{Q}_{\rm S}$ (as shown in in Fig. \ref{L2_0}(f)), which increases almost linearly with time. This is because $Q_{\rm S}$-oblivious offloads computing tasks to the MS regardless of the edge queue length. Besides, $Q_{\rm S}$-oblivious does not satisfy the power budget $c_{\rm th}$ at the MS (as shown in in Fig. \ref{L2_0}(g)). In contrast, as shown in Fig. \ref{L2_0}(e)-(h), ECO, LCO, and the proposed PLySE methods still produce stable data queues and satisfy the average energy consumption and energy causality constraints. 

\begin{figure}[tbp]
	\centering
	\includegraphics[scale=0.35]{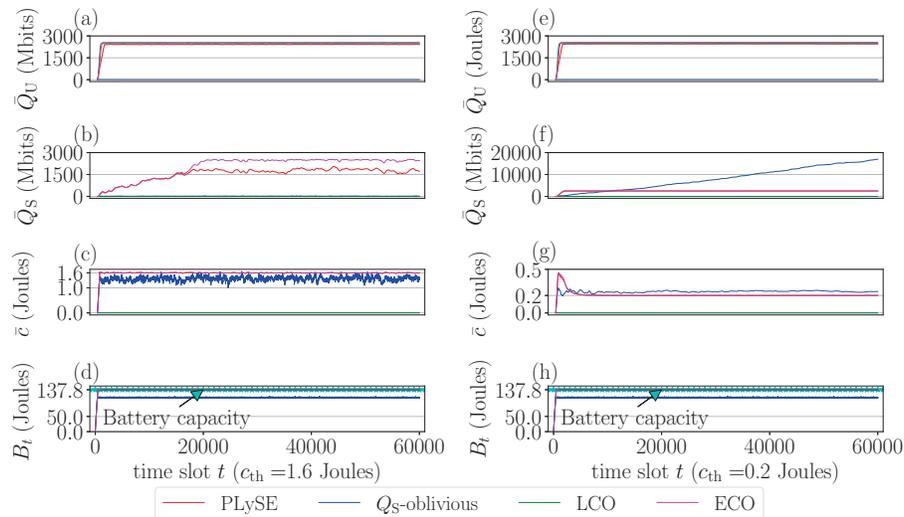}
	\captionsetup{font=footnotesize}
	\caption{Feasibility of PLySE and Benchmark methods.}
	\label{L2_0}
\end{figure}

\subsection{Impact of Lyapunov Control Parameter $V$}
In this subsection, we investigate the impact of parameter $V$ to the performance of PLySE. As shown in Fig. \ref{L1_0}, the data sensing rate $\bar{R}$ increases with $V$ and becomes saturated when $V$ is large enough (i.e., $V\geq 256\times10^7$ here). However, as $V$ increases, $\bar{Q}_{\rm S}$, $\bar{Q}_{\rm U}$ and $\Omega$ become large and grow rapidly especially when $V> 256\times 10^7$. These results are consistent with the theoretical analysis in section \ref{secV}, where a larger $V$ yields a larger battery capacity and longer data queues. By adjusting the value of $V$, PLySE offers a trade-off between the reduction of data queue length and increase of data sensing rate. In the following simulations, we set $V=256\times10^7$, whereby PLySE enjoys near-optimal data sensing rate with small data queue length and battery capacity.

\begin{figure}[tbp]
	\centering
	\includegraphics[scale=0.38]{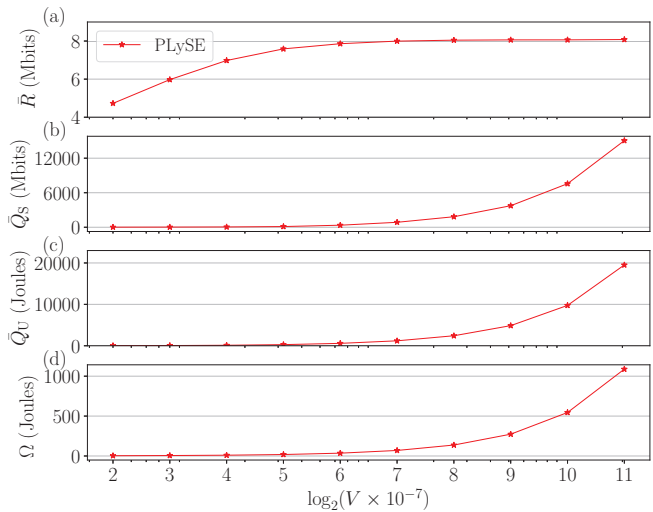}
	\captionsetup{font=footnotesize}
	\caption{System performance of PLySE under different $V$}
	\label{L1_0}
\end{figure}

\subsection{Impact of Primary Link Activities}
In this subsection, we investigate the impact of primary link activities on the data processing performance of PLySE. In Fig. \ref{L4}(a), we reveal the impact of interference tolerance at the PR to the task execution method. As shown in the figure, the local computing data size are almost constant when $\Gamma_{\rm th}$ varies, while the average offloading data size (i.e., the orange and green bars combined) increases with $\Gamma_{\rm th}$. This is because that a larger $\Gamma_{\rm th}$ allows a higher transmit power at the WD. Interestingly, as $\Gamma_{\rm th}$ increases, the WD offloads less task data to the MS when the primary link is idle (i.e., $a_t=0$), but more data when the primary link is active (i.e., $a_t=1$), because the WD has a larger freedom to control its transmit power when the primary link is active to reduce the overall energy consumption.


\begin{figure}
	\begin{minipage}{0.48\linewidth}
		\centerline{\includegraphics[scale=0.58]{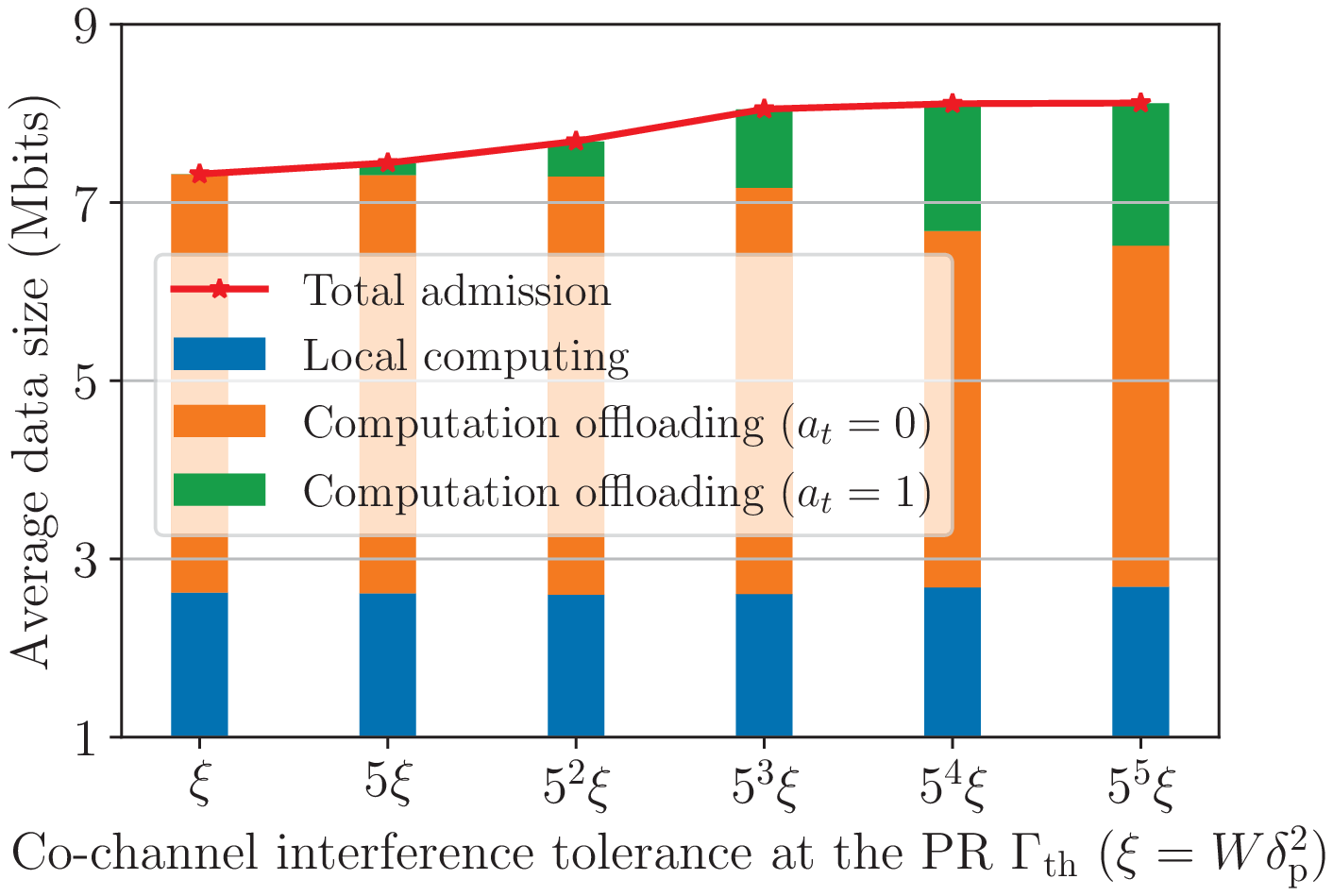}}
		\centerline{(a)}
	\end{minipage}
	\hfill
	\begin{minipage}{0.48\linewidth}
		\centerline{\includegraphics[scale=0.58]{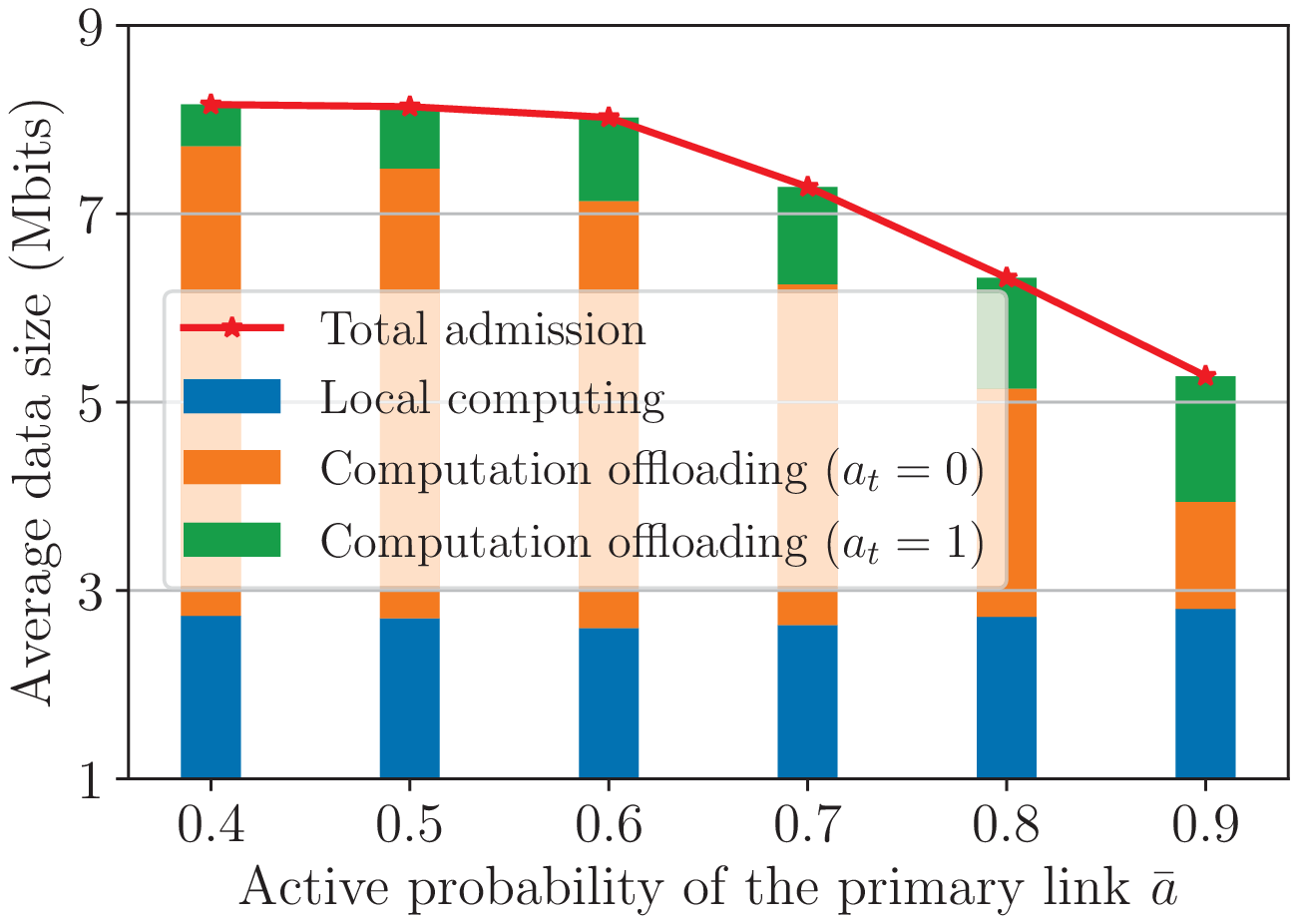}}
		\centerline{(b)}
	\end{minipage}
	\caption{The impact of (a) $\Gamma_{\rm th}$ and (b) $\bar{a}$ on computation offloading and local computing.}
	\label{L4}
\end{figure}

Then, we plot in Fig. \ref{L4}(b) the impact of active probability $\bar{a}$ of primary link on the computation offloading and local computing behavior of the WD. It displays that the local computing data size is almost immune to the change of $\bar{a}$. Besides, since higher $\bar{a}$ leaves the secondary MEC system less chance to access the licensed channel, the average offloading data size and total sensing data size decrease as $\bar{a}$ increases. Meanwhile, because the primary link is active for a larger portion of time under a larger $\bar{a}$, the total task data offloaded when the primary link is active (inactive) increases (decreases) with $\bar{a}$. 

\subsection{Performance Comparison Under Various System Parameters}
In order to show the effectiveness of the proposed PLySE method, we compare PLySE with the three benchmark methods in Fig. \ref{L3} considering various system parameters. For the $Q_{\rm S}$-oblivious method, the points of unstable data queue are omitted.

\begin{figure}
	\begin{minipage}{0.48\linewidth}
		\centerline{\includegraphics[scale=0.54]{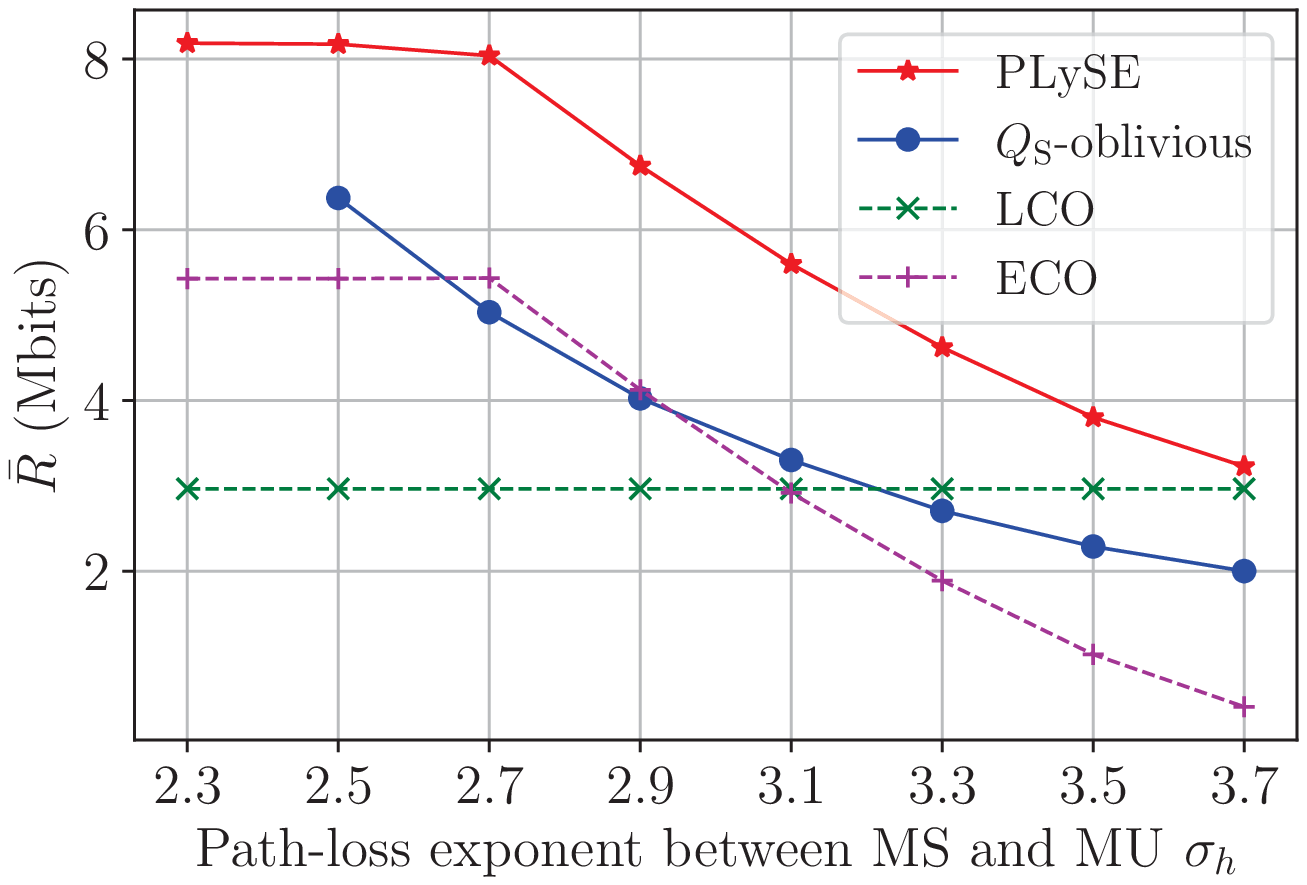}}
		\centerline{(a)}
	\end{minipage}
	\hfill
	\begin{minipage}{.48\linewidth}
		\centerline{\includegraphics[scale=0.54]{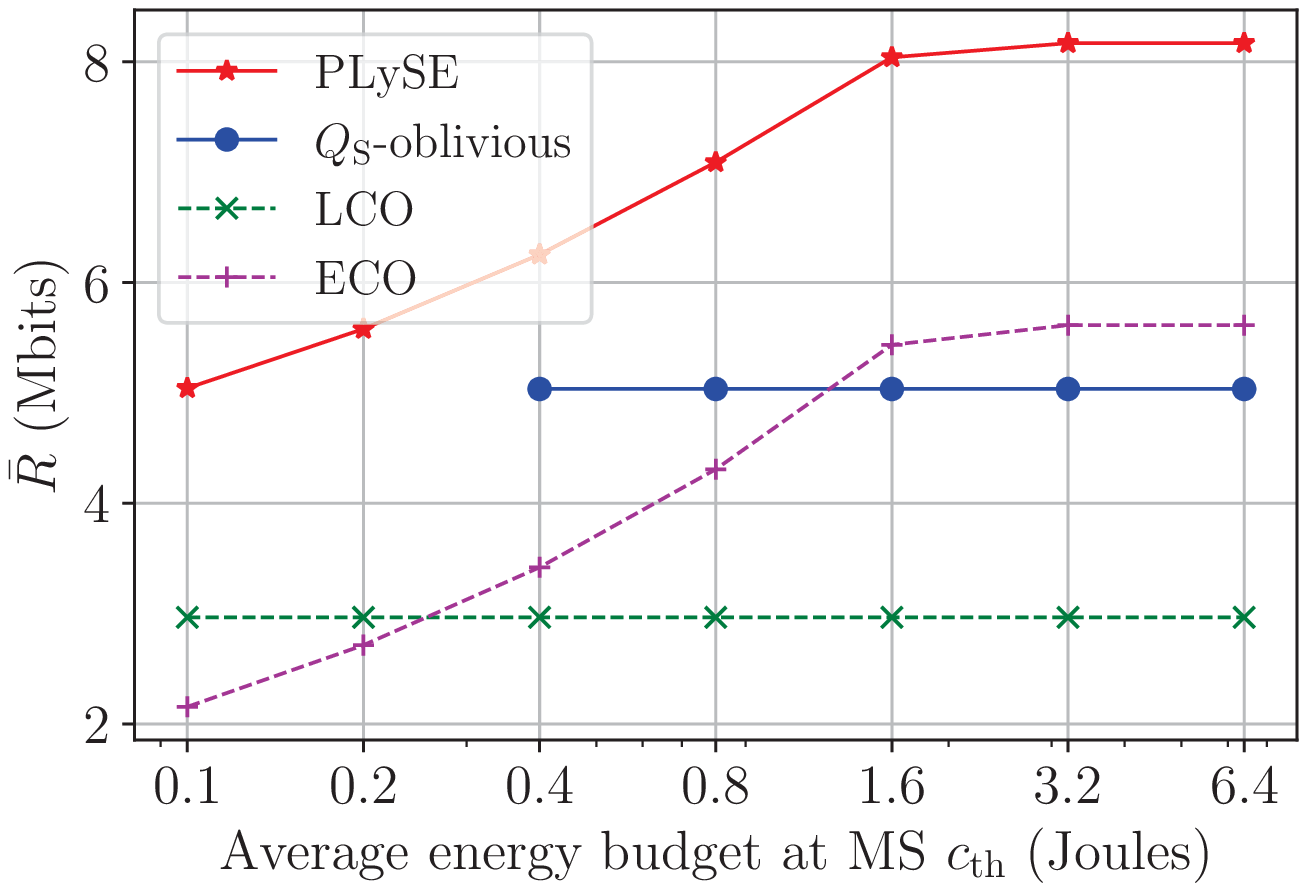}}
		\centerline{(b)}
	\end{minipage}
	\vfill
	\begin{minipage}{0.48\linewidth}
		\centerline{\includegraphics[scale=0.54]{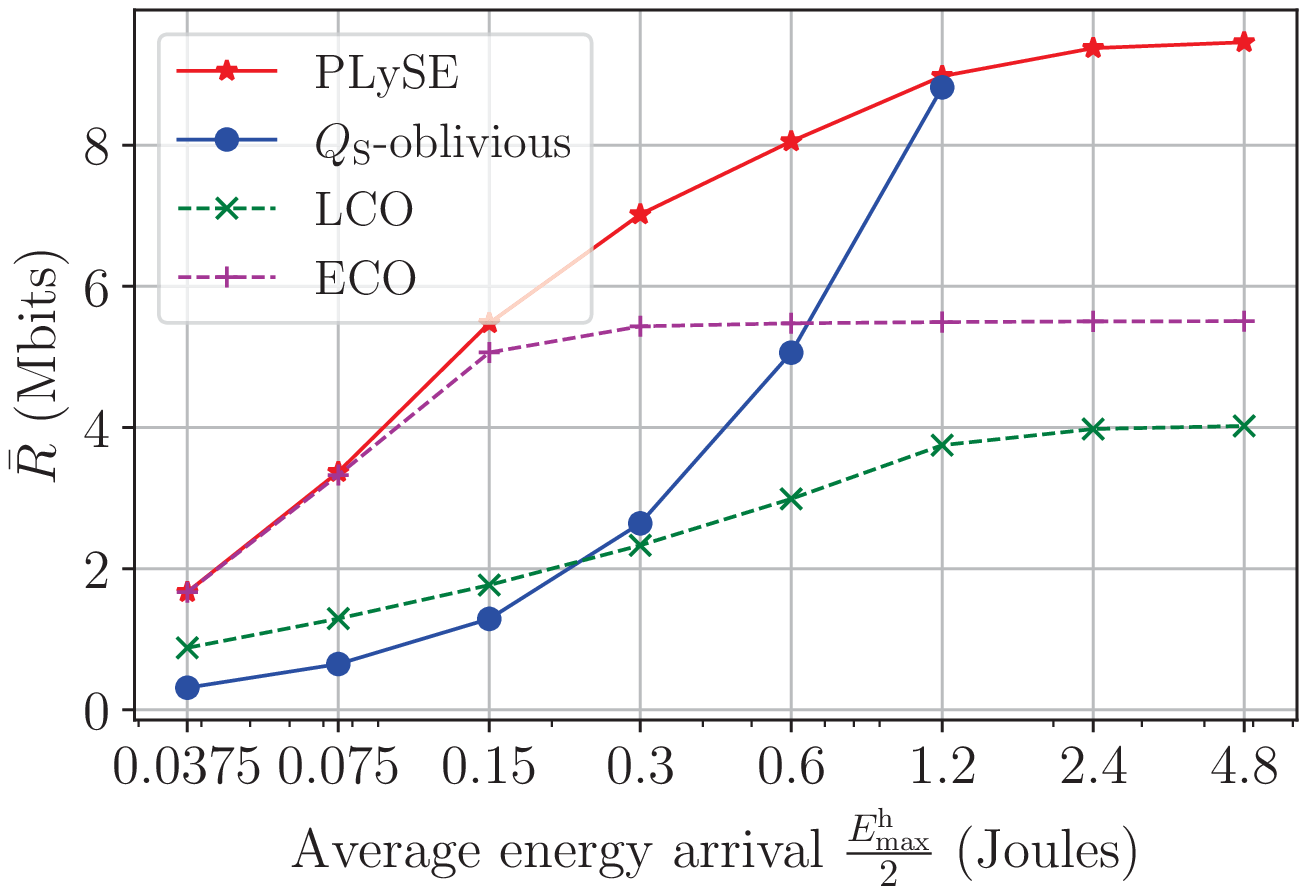}}
		\centerline{(c)}
	\end{minipage}
	\hfill
	\begin{minipage}{0.48\linewidth}
		\centerline{\includegraphics[scale=0.54]{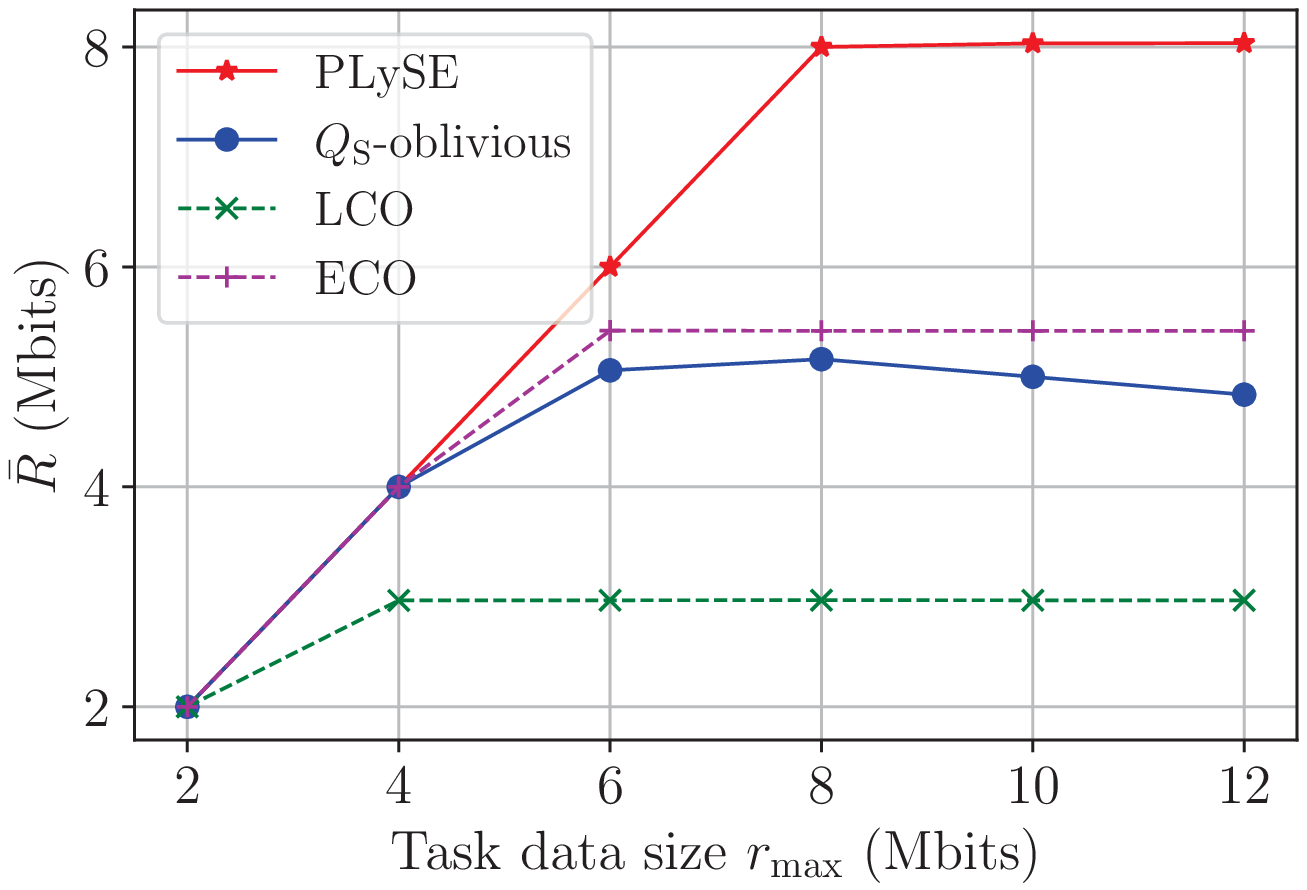}}
		\centerline{(d)}
	\end{minipage}
	\caption{Long-term average data sensing rate $\bar{R}$ versus: (a) path-loss exponent between the MS and WD $\sigma_{h}$; (b) power budget $c_{\rm th}$ at the MS; (c) average energy arrival rate $\frac{E_{\rm max}^{\rm h}}{2}$; and (d) maximum sensing data size $r_{\rm max}$.}
	\label{L3}
\end{figure}

We first examine in Fig. \ref{L3}(a) the performance of PLySE, ECO, LCO, and $Q_{\rm S}$-oblivious methods under different path-loss exponent $\sigma_{h}$. For PLySE, $Q_{\rm S}$-oblivious and ECO, $\bar{R}$ decreases with the rise of $\sigma_{h}$. This is for the reason that a larger $\sigma_{h}$ results in a more severe signal attenuation during computation offloading. In particular, at a small $\sigma_{h}=2.3$, the performance of PLySE and ECO become saturated due to the limited system resources (e.g., constrained transmit power and CPU frequency) and inherent co-channel interference from the primal link, while $Q_{\rm S}$-oblivious bears an unstable $Q_{\rm S}$ due to the surging offloading data at small $\sigma_{h}$. Because LCO does not perform task offloading, its data processing rate is not related to $\sigma_h$. As $\sigma_{h}$ increases, crossovers can be observed between ECO and LCO as well as $Q_{\rm S}$-oblivious and LCO. This is because that the WD tends to perform local computing at a large $\sigma_{h}$. Nevertheless, for all considered $\sigma_{h}$'s, the PLySE method shows a significant superiority over the other benchmarks.


In Fig. \ref{L3}(b), we depicts the long-term average data sensing rate as a function of the edge power budget $c_{\rm th}$. For PLySE and ECO, $\bar{R}$ grows with the increase of $c_{\rm th}$, and finally becomes steady due to the limited energy arrivals. Intuitively, LCO has constant data sensing performance unrelated to edge computing. For $Q_{\rm S}$-oblivious, it maintains the data queue stability and yields constant $\bar{R}$ when $c_{\rm th}\geq 0.4$ Joules. However, its offloading rate overwhelms the edge processing capability under stringent edge power constraint when $c_{\rm th}<0.4$, leading to infinite data queue backlog at the MS. Our proposed PLySE method enlarges its performance gaps to the other benchmarks as $c_{\rm th}$ climbs. Specifically, it offers 46.7\%, 132.8\%, 65.2\% higher average sensing rate than the $Q_{\rm S}$-oblivious, LCO, and ECO, respectively.


We also investigate in Fig. \ref{L3}(c) the impact of energy arrival rate on the average data sensing performance. Here, we vary the maximum energy arrival within a slot $E_{\rm max}^{\rm h}$, such that the average energy arrival rate is set as $\frac{E_{\rm max}^{\rm h}}{2}$ accordingly. It displays that the PLySE method provides significant performance improvement compared with the benchmarks for all $E_{\rm max}^{\rm h}$ considered. For all the methods, $\bar{R}$ increases with the rise of $E_{\rm max}^{\rm h}$. This is because larger $E_{\rm max}^{\rm h}$ allows higher local computing and task data offloading rate. Compared to the other three methods, $Q_{\rm S}$-oblivious is more sensitive to energy arrivals. In particular, $Q_{\rm S}$-oblivious yields the worst data sensing performance when $E_{\rm max}^{\rm h}=0.0375$ Joules, but rapidly improves $\bar{R}$ with increasing $E_{\rm max}^{\rm h}$ and finally achieve a similar $\bar{R}$ to PLySE at $E_{\rm max}^{\rm h}=1.2$ Joules. As $E_{\rm max}^{\rm h}$ grows, $Q_{\rm S}$-oblivious leads to unstable data queue at the MS, while PLySE, ECO and LCO achieve saturated $\bar{R}$ due to the hardware constraints on transmit power and CPU frequency.

%

In Fig. \ref{L3}(d), we further evaluate the data sensing performance under different task data size $r_{\rm max}$. For PLySE, LCO, and ECO, $\bar{R}$ grows with the increase of $r_{\rm max}$ and finally arrives at a stable value due to limited resources on data processing. On the other hand, with the growth of $r_{\rm max}$, $Q_{\rm S}$-oblivious first increases and then declines. This is due to the excessive energy consumption of $Q_{\rm S}$-oblivious on local computing and task data offloading, which leads to lower battery level under a larger $r_{\rm max}$. This renders a higher sensing cost $C_{\rm sen}$ and thus suppressing the data sensing rate. Nevertheless, for all $r_{\rm max}$'s, the proposed PLySE method significantly outperforms the three benchmark methods in terms of average data sensing rate.

\section{Conclusion}\label{sec7}
In this paper, we studied the optimal online policy design in a cognitive EH-MEC system. Aiming at maximizing the long-term average sensing rate while respecting the constraints on prescribed QoS requirement of primary link, long-term average power consumption at the mobile server and data queue stability, we developed an efficient online data sensing and processing algorithm named PLySE. In particular, we formulated a multi-stage stochastic optimization problem and transformed the intractable problem into per-slot deterministic optimization via the perturbed Lyapunov optimization technique. We further proposed low-complexity algorithm that obtains the optimal solution of the per-slot problem. Besides, we proved that the proposed PLySE algorithm achieves the optimal sensing rate asymptotically and meanwhile satisfying all the long-term constraints. The simulation results validated our analysis and demonstrated significant performance gain of the proposed PLySE algorithm over the considered benchmark methods. 

\appendices
\section{Proof of Lemma \ref{lem6}}\label{app0}
In the following, we derive an upper bound of \eqref{LyaDrift}. To begin with, we have that 
\begin{equation}\label{lem6_eq1}
\begin{split}
\frac{1}{2}\left(Q_{t+1}^{\rm U}\right)^2 - \frac{1}{2}\left(Q_{t}^{\rm U}\right)^2 &=  \frac{1}{2}\left(r_t-l_t^{\rm off}-l_t^{\rm loc}\right)^2 + Q_{t}^{\rm U}\left(r_t - l_t^{\rm off}-l_t^{\rm loc}\right)\\
&\leq \frac{1}{2}\left[\left(l_t^{\rm off}+l_t^{\rm loc}\right)^2+\left(r_t\right)^2\right]-Q_t^{\rm U}\left(l_t^{\rm off}+l_t^{\rm loc}-r_t\right).
\end{split}
\end{equation}
By taking the conditional expectation on both sides of \eqref{lem6_eq1}, we have 
\begin{equation}\label{lem6_eq2}
\begin{split}
\mathbb{E}\left[\frac{1}{2}\left(Q_{t+1}^{\rm U}\right)^2 - \frac{1}{2}\left(Q_{t}^{\rm U}\right)^2\mid\Theta_t\right]\leq D_1-\mathbb{E}\left[Q_t^{\rm U}\left(l_t^{\rm off}+l_t^{\rm loc}-r_t\right)\mid\Theta_t\right],
\end{split}
\end{equation}
Here, $D_1$ is a constant obtained as 
\begin{equation}\label{lem6_eq3}
\mathbb{E}\left[\frac{1}{2}\left[\left(l_t^{\rm off}+l_t^{\rm loc}\right)^2+\left(r_t\right)^2\right]\right]\leq \frac{1}{2}\left[\left(l_{\rm max}^{\rm off}+l_{\rm max}^{\rm loc}\right)^2+\left(r_{\rm max}\right)^2\right]\triangleq D_1,
\end{equation}
where $l_{\rm max}^{\rm off}=\mathbb{E}\left[WT\log_2\left(1+p_{\rm max}\gamma_t\right)\right]$ and  $l_{\rm max}^{\rm loc}=\frac{f^{\rm u}_{\rm max}T}{C}$. Similarly, we have that 
\begin{equation}\label{lem6_eq4}
\begin{split}
\mathbb{E}\left[\frac{1}{2}\left(Q_{t+1}^{\rm S}\right)^2 - \frac{1}{2}\left(Q_{t}^{\rm S}\right)^2\mid\Theta_t\right] &\leq D_2 - \mathbb{E}\left[Q_{t}^{\rm S}\left(l_t^{\rm edg}-l_t^{\rm off}\right)\mid\Theta_t\right],
\end{split}
\end{equation}
where $D_2=\frac{1}{2}\left[\left(l_{\rm max}^{\rm off}\right)^2 + \left(l_{\rm max}^{\rm edg}\right)^2\right]$ and $l_{\rm max}^{\rm edg}=\frac{f^{\rm s}_{\rm max}T}{C}$.

For the battery queue, we have that
\begin{equation}\label{lem6_eq5}
\begin{split}
\frac{1}{2}\tilde{B}_{t+1}^2-\frac{1}{2}\tilde{B}_t^2&\leq \frac{1}{2}\left(B_t-\lambda_{\rm e}e_t^{\rm u}+\lambda_{\rm e}e_t^{\rm h}-\Omega\right)^2 - \frac{1}{2}\left(B_{t}-\Omega\right)^2\\
&=\frac{1}{2}\lambda_{\rm e}^2\left(e_t^{\rm u}-e_t^{\rm h}\right)^2-\lambda_{\rm e}(B_t-\Omega)(e_t^{\rm u}-e_t^{\rm h})\\
&\leq \frac{1}{2}\left[\lambda_{\rm e}^2(e_t^{\rm u})^2+\lambda_{\rm e}^2\left(e_t^{\rm h}\right)^2\right]-\lambda_{\rm e}(B_t-\Omega)(e_t^{\rm u}-e_t^{\rm h}).
\end{split}
\end{equation}
By taking the conditional expectation on both sides of \eqref{lem6_eq1}, we have 
\begin{equation}\label{lem6_eq6}
\mathbb{E}\left[\frac{1}{2}\tilde{B}_{t+1}^2-\frac{1}{2}\tilde{B}_t^2\mid\Theta_t\right]\leq D_3 - \mathbb{E}\left[\lambda_{\rm e}(B_t-\Omega)(e_t^{\rm u}-e_t^{\rm h})\mid\Theta_t\right],
\end{equation}
where $D_3=\frac{1}{2}\left[(\lambda_{\rm e}e_{\rm max}^{\rm u})^2+\left(\lambda_{\rm e}e_{\rm max}^{\rm h}\right)^2\right]$ and $e_{\rm max}^{\rm u}=e_{\rm unit}^{\rm col}r_{\rm max} + p_{\rm max}T+\kappa_{\rm e}\left(f^{\rm u}_{\rm max}\right)^3T$.

On the other hand, using the fact that $\left[\max(x,0)\right]^2\leq(x)^2$, we have
\begin{equation}\label{lem6_eq7}
\begin{split}
\frac{1}{2}Z_{t+1}^2-\frac{1}{2}Z_t^2&\leq \frac{1}{2}\left(\lambda_{\rm c} e_t^{\rm edg}-\lambda_{\rm c} c_{\rm th}\right)^2 + Z_t\lambda_{\rm c}\left(e_t^{\rm edg}-c_{\rm th}\right)\\
\leq& \frac{1}{2}\left(\lambda_{\rm c} e_t^{\rm edg}\right)^2+\frac{1}{2}\left(\lambda_{\rm c} c_{\rm th}\right)^2 + Z_t\lambda_{\rm c}\left(e_t^{\rm edg}-c_{\rm th}\right)
\end{split}
\end{equation}
Correspondingly, we have 
\begin{equation}\label{lem6_eq8}
\mathbb{E}\left[\frac{1}{2}Z_{t+1}^2-\frac{1}{2}Z_t^2\mid\Theta_t\right]\leq D_4 - \mathbb{E}\left[ Z_t\lambda_{\rm c}\left(c_{\rm th}-e_t^{\rm edg}\right)\mid\Theta_t\right],
\end{equation}
where $D_4=\frac{1}{2}\left[\left(\lambda_{\rm c}e_{\rm max}^{\rm edg}\right)^2+\left(\lambda_{\rm c}c_{\rm th}\right)^2\right]$ and $e_{\rm max}^{\rm edg}=\kappa_{\rm c}\left(f_{\rm max}^{\rm s}\right)^3T$.

By summing up \eqref{lem6_eq2}, \eqref{lem6_eq4}, \eqref{lem6_eq6} and \eqref{lem6_eq8}, we obtain an upper bound of $\Delta^t_{V}$ in \eqref{LyaDriPlusPen}.

\section{Proof of Proposition \ref{lem_opt_solution}}\label{App_Opt_solution}
Depending on the value of $Q_t^{\rm U}-Q_t^{\rm S}$, we solve \eqref{Prob_Per_slot_sim} by considering the following two cases:

\textbf{Case I}: $Q_t^{\rm U}-Q_t^{\rm S}\geq0$. Remind that $\tilde{B}_t=B_t-\Omega \leq 0$. In this case, $F\left(f_t^{\rm u}\right)$ and $G\left(p_t^{\rm u}\right)$ are concave functions and achieve maximum at $\tilde{f}_t^{\rm u}=\sqrt{\frac{-Q_t^{\rm u}}{3\lambda_{\rm e}\tilde{B}_t\kappa_{\rm c}C}}$ and $\tilde{p}_t^{\rm u}=\frac{\left(Q_t^{\rm S}-Q_t^{\rm U}\right)W}{\lambda_{\rm e}\tilde{B}_t\ln2}-\frac{1}{\gamma_t}$, respectively. In a special case of $\tilde{B}_t = 0$, we set $\tilde{f}_t^{\rm u} = f_{\rm max}^{\rm u}$ and $\tilde{p}_t^{\rm u}=p_{\rm th}$. Denote $\hat{f}_t^{\rm u}=\min\left(\tilde{f}_t^{\rm u}, \bar{f}_{\rm th}\right)$ and $\hat{p}_t^{\rm u}=\min\left(\max\left(\tilde{p}_t^{\rm u},0\right), \bar{p}_{\rm th}\right)$. Then, the optimal solution must satisfy that $f_t^{\rm u\ast}\in\left[0,\hat{f}_t^{\rm u}\right]$ and $p_t^{\rm u\ast}\in\left[0,\hat{p}_t^{\rm u}\right]$, where both $F\left(f_t^{\rm u}\right)$ and $G\left(p_t^{\rm u}\right)$ are monotonically increasing. Let $\hat{l}_{\rm max}^{\rm loc}=\frac{\hat{f}_t^{\rm u}T}{C}$ and $\hat{l}_{\rm max}^{\rm off}=WT\log_2\left(1+\hat{p}_t^{\rm u}\gamma_t\right)$ denote the maximum amount of data processed via local computing and computation offloading in time slot $t$, respectively. Then, the optimal solution of \eqref{Prob_Per_slot_sim} can be obtained as	below:
	\begin{itemize}
		\item When $\hat{l}_{\rm max}^{\rm off}+\hat{l}_{\rm max}^{\rm loc} \leq Q_t^{\rm U}$, we can directly obtain that $f_t^{\rm u\ast}=\hat{f}_t^{\rm u}$ and $p_t^{\rm u\ast}=\hat{p}_t^{\rm u}$. 
		\item When $\hat{l}_{\rm max}^{\rm off}+\hat{l}_{\rm max}^{\rm loc} > Q_t^{\rm U}$, $l_t^{\rm off}+l_t^{\rm loc} = Q_t^{\rm U}$ must hold at optimum. By substituting $p_t^{\rm u}=\mathcal{F}_p(f_t^{\rm u})$ into $G(p_t^{\rm u})$, we can rewrite \eqref{Prob_Per_slot_sim} as
		\begin{subequations}\label{Prob_Per_slot_sim_case3}
			\begin{align}
			\underset{\substack{f_t^{\rm u}}} \max~& U\left(f_t^{\rm u}\right),~~\st
			~~f_{\rm I}^{\rm lb} \leq\! f_t^{\rm u} \!\leq\! f_{\rm I}^{\rm ub},
			\end{align}
		\end{subequations}
		where $U\left(f_t^{\rm u}\right)=\lambda_{\rm e}\tilde{B}_t\left[\mathcal{F}_p\left(f_t^{\rm u}\right)T+\kappa_{\rm c}\left(f_t^{\rm u}\right)^3T\right]\!+\!\left(Q_t^{\rm U}\right)^2 + \frac{Q_t^{\rm S}f_t^{\rm u}T}{C} - Q_t^{\rm S}Q_t^{\rm U}$, $f_{\rm I}^{\rm ub} = \hat{f}_t^{\rm u}$, and $f_{\rm I}^{\rm lb} = \max\left(0, \mathcal{F}_f\left(\hat{p}_t^{\rm u}\right)\right)$. When $\tilde{B}_t= 0$, $U\left(f_t^{\rm u}\right)=\!\left(Q_t^{\rm U}\right)^2 + \frac{Q_t^{\rm S}f_t^{\rm u}T}{C} - Q_t^{\rm S}Q_t^{\rm U}$ is a linear function of $f_t^{\rm u}$. The optimal solution of \eqref{Prob_Per_slot_sim_case3} can be obtained as $f_t^{\rm u\ast}=\hat{f}_t^{\rm u}$ and thus $p_t^{\rm u\ast}=\mathcal{F}_p\left(\hat{f}_t^{\rm u}\right)$. When $\tilde{B}_t< 0$, $U\left(f_t^{\rm u}\right)$ is a concave function. We assume $U\left(f_t^{\rm u}\right)$ achieves maximum at $\bar{f}_t^{\rm u}$, which can be obtained by solving equation
		\begin{equation}
		U^\prime(f_t^{\rm u}) = \frac{\partial U}{\partial f_t^{\rm u}} = 3\lambda_{\rm e}\kappa_{\rm c}T\tilde{B}_t\left(f_t^{\rm u}\right)^2 - \lambda_{\rm e}\tilde{B}_tTA_02^{-\frac{f_t^{\rm u}}{WC}} + \frac{T}{C}Q_t^{\rm S} = 0,
		\end{equation}
		where $A_0 = \frac{\ln 2}{WC\gamma_t}2^{\frac{Q_t^{\rm U}}{WT}}$.
		\begin{Lemma}\label{LemCase3}
			$U'(f_t^{\rm u})$ is a monotonically decreasing function of $f_t^{\rm u}$ and $U'(f_t^{\rm u})=0$ has a unique solution $\bar{f}_t^{\rm u}\in[0,+\infty)$.
		\end{Lemma}
		\begin{proof}
			By taking the derivative of $U'$ in terms of $f_t^{\rm u}$, we can easily find that $\frac{\partial U^\prime}{\partial f_t^{\rm u}}<0$ for $f_t^{\rm u}\in[0,+\infty)$. That is, $U^\prime(f_t^{\rm u})$ is a monotonically decreasing function of $f_t^{\rm u}$. When $f_t^{\rm u}=0$, $U^\prime(f_t^{\rm u})=-\lambda_{\rm e}TA_0\tilde{B}_t + \frac{T}{C}Q_t^{\rm S}>0$. When $f_t^{\rm u}\to +\infty$, $U'(f_t^{\rm u})\to -\infty$. Therefore, there is a unique solution $\bar{f}_t^{\rm u}\in[0,+\infty)$ for $U^\prime(f_t^{\rm u})=0$.   
		\end{proof}
		Based on Lemma \ref{LemCase3}, we can obtain $\bar{f}_t^{\rm u}$ that satisfies $U'(\bar{f}_t^{\rm u})=0$ using bi-section search method. Then, for the case of $\tilde{B}_t< 0$, the optimal solution of \eqref{Prob_Per_slot_sim_case3} can be given as $f_t^{\rm u\ast}=\breve{f}_t^{\rm u}$, where $\breve{f}_t^{\rm u}\!=\!\min\left(\max\left(f_{\rm I}^{\rm lb}\!,\!\bar{f}_t^{\rm u}\right)\!,\!f_{\rm I}^{\rm ub}\right)$. Correspondingly, $p_t^{\rm u\ast}=\mathcal{F}_p\left(\breve{f}_t^{\rm u}\right)$.
	\end{itemize}

	\textbf{Case II}: $Q_t^{\rm U}-Q_t^{\rm S}<0$. In this case, $F$ is a concave function of $f_t^{\rm u}$ and $G$ monotonically decreases with $p_t^{\rm u}$. The optimal solution of \eqref{Prob_Per_slot_sim} can be easily obtained as $p_t^{\rm u\ast}=0$ and $f_t^{\rm u\ast}=\hat{f}_t^{\rm u}$.
	
	By summarizing the results in Case I and II, we finally obtain the optimal solution of \eqref{Prob_Per_slot_sim} as shown in \eqref{Opt_pu_fu}.

\vspace{-1em}
\section{Proof of Proposition \ref{lem3}}\label{app1}
The main idea of proving Proposition \ref{lem3} is to select an appropriate $\Omega$ to make constraint \eqref{EH_caus} implicit for any $B_t\in[0,\Omega]$. By initially setting $\Omega\geq \lambda_{\rm e}e_{\rm max}^{\rm u}+\lambda_{\rm e}E_{\rm max}^{\rm h}$, we prove Proposition \ref{lem3} by considering following three cases depending on the value of $B_t$. 

\textbf{Case I}: When $B_t\in[\lambda_{\rm e}e_{\rm max}^{\rm u}, \Omega]$, we have $B_{t+1} \leq \min\left(\Omega+\lambda_{\rm e}E_{\rm max}^{\rm h},\Omega\right) = \Omega$ based on the energy dynamic \eqref{E_evol}. Because $B_t\geq \lambda_{\rm e}e_{\rm max}^{\rm u}\geq \lambda_{\rm e}e_t^{\rm u}$ for all feasible $r_t$, $f_t^{\rm u}$ and $p_t^{\rm u}$, the energy causality constraint \eqref{EH_caus} is satisfied. Thus, we have $0\leq B_{t+1}\leq \Omega$.
	
\textbf{Case II}: When $B_t\in[0, B_{\rm min}]$, we have $B_{t+1} \leq B_t + \lambda_{\rm e}e^{\rm h}_t \leq B_{\rm min} + \lambda_{\rm e}E_{\rm max}^{\rm h}< \Omega$. From \eqref{EH_caus}, the WD stops consuming energy on data sensing and processing, i.e., $e_t^{\rm col}=e_t^{\rm off}=e_t^{\rm loc}=0$. The energy causality constraint \eqref{EH_caus} is satisfied and thus we have $0\leq B_{t+1}< \Omega$.
	
\textbf{Case III}: When $B_t\in[B_{\rm min}, \lambda_{\rm e}e_{\rm max}^{\rm u}]$, we have $B_{t+1} \leq \lambda_{\rm e}e_{\rm max}^{\rm u}+\lambda_{\rm e}E_{\rm max}^{\rm h}\leq \Omega$. To satisfy the energy causality $\lambda_{\rm e}e_t^{\rm u}\leq B_{t}$ for all $B_t\in[B_{\rm min}, \lambda_{\rm e}e_{\rm max}^{\rm u}]$, one possible solution is to set an appropriate $\Omega$, such that the energy cost on data collection is zero (i.e., $C_{\rm sen}>0$ and thus $r_t=0$) and that on data transmission and local computing is less than $B_{\rm min}$. Accordingly, we derive the condition on $\Omega$ as follows.
	\begin{itemize}
		\item Based on \eqref{Opt_r}, we have that $r_t=0$ if $\frac{V-Q_t^{\rm U}}{\lambda_{\rm e}e_{\rm unit}^{\rm col}}+B_t<\frac{V}{\lambda_{\rm e}e_{\rm unit}^{\rm col}}+\lambda_{\rm e}e_{\rm max}^{\rm u}\leq\Omega$, where the first inequality holds because $B_t \in [B_{\rm min}, \lambda_{\rm e}e_{\rm max}^{\rm u}]$ by assumption.
		
		\item Based on \eqref{Opt_pu_fu}, we observe that the WD consumes the maximum energy on data transmission and local computing when $\left\{f_t^{\rm u\ast},p_t^{\rm u\ast}\right\} = \left\{\hat{f}_t^{\rm u}, \hat{p}_t^{\rm u}\right\}$. To ensure $\lambda_{\rm e}\left(e_t^{\rm off}+e_t^{\rm loc}\right)\leq B_{\rm min}$, it suffices to select an $\Omega$ that satisfies $\lambda_{\rm e}\left(\hat{e}_t^{\rm off}+\hat{e}_t^{\rm loc}\right)\leq B_{\rm min}$, where $\hat{e}_t^{\rm off} = \hat{p}_t^{\rm u}T$ and $\hat{e}_t^{\rm loc} = \kappa_{\rm e}\left(\hat{f}_t^{\rm u}\right)^3T$. Recall that $\hat{p}_t^{\rm u} = \min\left(\max\left(\tilde{p}_t^{\rm u},0\right),\bar{p}_{\rm th}\right)$ and $\hat{f}_t^{\rm u}=\min\left(\tilde{f}_t^{\rm u},\bar{f}_{\rm th}\right)$, where $\tilde{f}_t^{\rm u}=\sqrt{\frac{-Q_t^{\rm u}}{3\lambda_{\rm e}\tilde{B}_t\kappa_{\rm c}C}}$ and $\tilde{p}_t^{\rm u}=\frac{\left(Q_t^{\rm S}-Q_t^{\rm U}\right)W}{\lambda_{\rm e}\tilde{B}_t\ln2}-\frac{1}{\gamma_t}$, respectively. In the following, we discuss the value of $\Omega$ considering two sub-cases: 
		
		a) When $\tilde{p}_t^{\rm u}> 0$, we have $\hat{p}_t^{\rm u}\leq \tilde{p}_t^{\rm u}$. Using the result in Lemma \ref{lem_ub} that $Q_t^{\rm U}\leq V+r_{\rm max}$, we have 
		\begin{subequations}
		\setlength{\abovedisplayskip}{5pt}
		\setlength{\belowdisplayskip}{5pt}
		\begin{align}
		\hat{f}_t^{\rm u} &\leq \sqrt{\frac{-Q_t^{\rm u}}{3\lambda_{\rm e}\tilde{B}_t\kappa_{\rm c}C}} \leq \sqrt{\frac{V+r_{\rm max}}{3\lambda_{\rm e}\left(\Omega-\lambda_{\rm e}e_{\rm max}^{\rm u}\right)\kappa_{\rm c}C}},\label{App1_eq1}\\
		\hat{p}_t^{\rm u} &\leq \frac{\left(Q_t^{\rm S}-Q_t^{\rm U}\right)W}{\lambda_{\rm e}\tilde{B}_t\ln2}-\frac{1}{\gamma_t}\leq \frac{\left(V+r_{\rm max}\right)W}{\lambda_{\rm e}\left(\Omega-\lambda_{\rm e}e_{\rm max}^{\rm u}\right)\ln 2},\label{App1_eq2}
		\end{align}
		\end{subequations}
		where the first inequality of \eqref{App1_eq1} holds because $\hat{f}_t^{\rm u}\leq \tilde{f}_t^{\rm u}$. By submitting the right-hand sides of \eqref{App1_eq1} and \eqref{App1_eq2} into $\lambda_{\rm e}\left(\hat{e}_t^{\rm off}+\hat{e}_t^{\rm loc}\right)\leq B_{\rm min}$, the inequality can be equivalently written as
		\begin{equation}\label{lem4_case3_eq1}
		\begin{split}
		H(\Omega) = \lambda_{\rm e}\kappa_{\rm e}\left(\sqrt{\frac{V+r_{\rm max}}{3\lambda_{\rm e}\left(\Omega-\lambda_{\rm e}e_{\rm max}^{\rm u}\right)\kappa_{\rm c}C}}\right)^3T + \frac{\left(V+r_{\rm max}\right)W}{\left(\Omega-\lambda_{\rm e}e_{\rm max}^{\rm u}\right)\ln 2}T \leq B_{\rm min}.
		\end{split}
		\end{equation}
		Notice that $\Omega\geq \lambda_{\rm e}e_{\rm max}^{\rm u}$. It is obviously that $H(\Omega)$ is decreasing with $\Omega$. Besides, $H(\Omega)\to +\infty$ when $\Omega\to \lambda_{\rm e}e_{\rm max}^{\rm u}$ and $H(\Omega)\to 0$ when $\Omega\to +\infty$. Therefore, there exists at least one $\Omega\in\left(\lambda_{\rm e}e_{\rm max}^{\rm u},+\infty\right)$ that satisfies \eqref{lem4_case3_eq1}. After some simple mathematical manipulations, we can rewrite \eqref{lem4_case3_eq1} as
		\begin{equation}\label{lem4_case3_eq2}
		\bar{H}(x) = A_1^2x^3 + 2A_1A_2x^2 + A_2^2x + A_3 \geq 0
		\end{equation}
		where $x=\Omega-\lambda_{\rm e}e_{\rm max}^{\rm u}$, $A_1 = \frac{3C\kappa_{\rm c}B_{\rm min}}{\kappa_{\rm e}\left(V+r_{\rm max}\right)T}$, $A_2 = -\frac{3CW\kappa_{\rm c}}{\kappa_{\rm e}\ln 2}$, $A_3 = -\frac{V+r_{\rm max}}{3\lambda_{\rm e}C\kappa_{\rm c}}$. $\bar{H}(x)$ is a cubic function and the three solutions of $\bar{H}(x)=0$ can be given as
		\begin{equation}
		x_k = 2\sqrt{-\frac{\bar{A}_1}{3}}\cos\left[\frac{1}{3}\arccos\left(\frac{3\bar{A}_2}{2\bar{A}_1}\sqrt{-\frac{3}{\bar{A}_1}}\right)-\frac{2\pi k}{3}\right], ~\text{for}~k=0,1,2,
		\end{equation}
		Where $\bar{A}_1 = -\frac{A_2^2}{3A_1^2}$ and $\bar{A}_2 = \frac{-2A_2^3+27A_1A_3}{27A_1^3}$. 
		
		Notice that $\bar{A}_1<0$ and the coefficient of $x^3$ in $\bar{H}(x)$ is $A_1^2>0$. Accordingly, we can satisfy \eqref{lem4_case3_eq2} by selecting  $x\geq x_{\rm max}$, i.e., $\Omega\geq x_{\rm max}+\lambda_{\rm e}e_{\rm max}^{\rm u}$, where $x_{\rm max} = \max_{k}\left(x_k\right)$. 
		
		b) When $\tilde{p}_t\leq 0$, we have $\hat{p}_t=0$. In this sub-case, we need to select an $\Omega$ satisfying $\hat{e}_t^{\rm loc}\leq B_{\rm min}$. That is,
		\begin{equation}\label{APP1_eq3}
		\lambda_{\rm e}\kappa_{\rm e}\left(\sqrt{\frac{V+r_{\rm max}}{3\lambda_{\rm e}\left(\Omega-\lambda_{\rm e}e_{\rm max}^{\rm u}\right)\kappa_{\rm c}C}}\right)^3T \leq B_{\rm min}.
		\end{equation}
		Obviously, $\Omega\geq x_{\rm max}+\lambda_{\rm e}e_{\rm max}^{\rm u}$ also satisfies \eqref{APP1_eq3}.
	\end{itemize}
From the above discussions, we can set $\Omega\!\!\geq\!\max\left(\!\frac{V}{\lambda_{\rm e}e_{\rm unit}^{\rm col}}\!+\!\lambda_{\rm e}e_{\rm max}^{\rm u},\!x_{\rm max}\!+\!\lambda_{\rm e}e_{\rm max}^{\rm u}\!\right)\!\!+\!\lambda_{\rm e}E_{\rm max}^{\rm h}$ to meet the energy causality constraint \eqref{EH_caus} for all $B_t\!\in\![0,\Omega]$, which ends the proof of Proposition \ref{lem3}.

\section{Proof of Proposition \ref{lem4}}\label{app2}
To start with, we introduce the following two useful lemma to prove Proposition \ref{lem4}.
\begin{Lemma}\label{lem2}
	The optimal utility $\bar{R}_{\rm P2}^\ast$ to the relaxed problem (P2) can be achieved arbitrarily closely by an $\omega$-only policy, i.e., for any $\delta>0$, there exists an $\omega$-only policy $\Pi$, achieves
	\begin{equation}
	\mathbb{E}\left[r_t^{\Pi}\right] \geq l_{\rm P2}^\ast - \delta,
	\end{equation}
	while satisfying the constraints \eqref{SNR_thre}, \eqref{ledge_off_cons}, and \eqref{Prob_const2} in (P2), and
	\begin{subequations}\label{AppD_eq2}
		\setlength{\abovedisplayskip}{5pt}
		\setlength{\belowdisplayskip}{5pt}
		\begin{align}
		&\mathbb{E}\left[e_t^{\rm edg,\Pi}-c_{\rm th}\right] \leq \delta,~ \mathbb{E}\left[e_t^{\rm u, \Pi}-e_t^{\rm h,\Pi}\right] \leq \delta,\\
		&\mathbb{E}\left[l_t^{\rm{off},\Pi}\right] \leq \mathbb{E}\left[l_t^{\rm{edg},\Pi}\right]+\delta,~
		\mathbb{E}\left[r_t^{\Pi}\right] \leq \mathbb{E}\left[l_t^{\rm{off},\Pi}+l_t^{\rm{loc},\Pi}\right] + \delta.
		\end{align}
	\end{subequations}
\end{Lemma}

\begin{proof}
	The proof follows the framework of Theorem 4.5 in \cite{Neely2010} and is omitted here for brevity.
\end{proof}

\vspace{-1em}
\begin{Lemma}\label{lem7}
	If $Z_t$ is mean rate stable, i.e., $\lim_{N\to\infty}\frac{\mathbb{E}[Z_N]}{N}=0$, then the average constraint \eqref{Bud_cons} is satisfied.
\end{Lemma}
\begin{proof}
	Using the sample path property (Lemma 2.1 in \cite{Neely2010}), we have that
	\begin{equation}
	\frac{Z_{N}}{N}-\frac{Z_{0}}{N}\geq \frac{1}{N}\mathsmaller{\sum}_{t=1}^{N}\left(e_t^{\rm edg}-c_{\rm th}\right).
	\end{equation}
	Dropping the negative terms $\frac{Z_0}{N}$, taking the expectation of the equation above and letting $N\to\infty$, we have that
	\begin{equation}
	\lim_{N\to\infty}\frac{\mathbb{E}\left[Z_{N}\right]}{N}\geq\lim_{N\to\infty}\frac{1}{N}\mathsmaller{\sum}_{t=1}^{N}\mathbb{E}\left[e_t^{\rm edg}-c_{\rm th}\right].
	\end{equation}
	Submitting $\lim_{N\to\infty}\frac{\mathbb{E}[Z_N]}{N}=0$, then we have $\lim_{N\to\infty}\frac{1}{N}\sum_{t=1}^{N}\mathbb{E}\left[e_t^{\rm edg}\right]\leq c_{\rm th}$, which completes the proof.
\end{proof}
\textbf{Proof of Proposition \ref{lem4}:} Consider the upper bound on the Lyapunov drift-plus-penalty function \eqref{LyaDriPlusPen}. We denote the policy produced by PLySE as $\Psi$. Since the solution of PLySE minimizes the upper bound on the Lyapunov drift-plus-penalty function $\Delta^{t}_V$, the following inequality holds:
\begin{small}
\begin{equation}\label{Lem_eq1}
\begin{split}
&\Delta^{t}_V=\Delta^t-V\mathbb{E}\left[r_t|\Theta_t\right]\\
&\leq D \!\!-\!\!\mathbb{E}\left\{\!\!Vr_t^{\Psi}\!\!+\!\!Z_t\lambda_{\rm c}\left(\!\!c_{\rm th}\!\!-\!e_t^{\rm edg,\Psi}\!\right)\!\!+\!\!\lambda_{\rm e}(B_t\!\!-\!\!\Omega)\left(\!\!e_t^{\rm u,\Psi}\!\!\!-\!\!\!e_t^{\rm h,\Psi}\!\right)\!\!+\!\! Q_t^{\rm U}\left(\!l_t^{\rm off,\Psi}\!\!+\!\!l_t^{\rm loc,\Psi}\!\!\!-\!\!r_t^{\Psi}\!\right)\!\!+\!\!Q_t^{\rm S}\left(\!l_t^{\rm edg,\Psi}\!\!\!-\!\!l_t^{\rm off,\Psi}\!\right)\!\!\mid\!\Theta_t\!\right\}\\
&\leq D \!\!-\!\!\mathbb{E}\left\{\!\!Vr_t^{\Pi}\!\!+\!\!Z_t\lambda_{\rm c}\left(\!\!c_{\rm th}\!\!-\!e_t^{\rm edg,\Pi}\!\right)\!\!+\!\!\lambda_{\rm e}(B_t\!\!-\!\!\Omega)\left(\!\!e_t^{\rm u,\Pi}\!\!\!-\!\!\!e_t^{\rm h,\Pi}\!\right)\!\!+\!\! Q_t^{\rm U}\left(\!l_t^{\rm off,\Pi}\!\!+\!\!l_t^{\rm loc,\Pi}\!\!\!-\!\!r_t^{\Pi}\!\right)\!\!+\!\!Q_t^{\rm S}\left(\!l_t^{\rm edg,\Pi}\!\!\!-\!\!l_t^{\rm off,\Pi}\!\right)\!\mid\!\Theta_t\!\right\}\\
&\overset{(\dag)}= D - \mathbb{E}\left[Vr_t^{\Pi}\right] \!+ \!\mathbb{E}\left[Z_t\lambda_{\rm c}\left(e_t^{\rm edg,\Pi}\!-\!c_{\rm th}\right)\right]\!+\!\mathbb{E}\left[\lambda_{\rm e}(\Omega\!-\!B_t)(e_t^{\rm u,\Pi}\!-\!e_t^{\rm h,\Pi})\right]\\
& + \mathbb{E}\left[Q_t^{\rm U}\left(r_t^{\Pi}-l_t^{\rm off,\Pi}-l_t^{\rm loc,\Pi}\right)\right]+\!\mathbb{E}\left[Q_t^{\rm S}(l_t^{\rm off,\Pi}-l_t^{\rm edg,\Pi})\right]\\
&\overset{(\ddag)}\leq D - V\left(\bar{R}_{\rm P2}^\ast-\delta\right) + \left[Z_t\lambda_{\rm c}+ Q_t^{\rm S} + \lambda_{\rm e}\Omega + Q_{\rm max}^{\rm U}\right]\delta,
\end{split}
\end{equation}
\end{small}where $(\dag)$ holds for the independence of policy $\Pi$ on $\Theta_t$ and $(\ddag)$ is obtained by plugging \eqref{AppD_eq2}. Let $\delta \to 0$, we have that
\begin{equation}\label{Eq1_app2}
\Delta^t-V\mathbb{E}\left[r_t^{\Psi}|\Theta_t\right]\leq D-V\bar{R}_{\rm P2}^\ast.
\end{equation}
By summing up the both sides of \eqref{Eq1_app2} from $t=1$ to $N$, and taking iterated expectations and telescoping sums, then normalizing by $VN$, we have that 
\begin{equation}
\frac{\mathbb{E}\left[\Phi_{t}\right]\!-\!\mathbb{E}\left[\Phi_{0}\right]}{NV}\!-\!\frac{1}{N}\mathsmaller{\sum}_{t=0}^{N-1}\mathbb{E}\left[r_t^{\Psi}\right]\!\leq\! \frac{D}{V}\!-\!\bar{R}_{\rm P2}^\ast.
\end{equation}
By rearranging terms and letting $N \to \infty$, we prove a) that 
\begin{equation}
\bar{R}_{\Psi}= \lim_{N\!\to+\!\infty}\!\!\frac{1}{N}\!\mathsmaller{\sum}_{t=0}^{N\!-\!1}\mathbb{E}\left[r_t^{\Psi}\right]\!\!\geq \!\bar{R}_{\rm P2}^\ast\!-\!\frac{D}{V} \!\overset{(\S)}\geq\! \bar{R}_{\rm P1}^\ast\!-\!\frac{D}{V},
\end{equation}
where $\S$ is due to $\bar{R}_{\rm P1}^\ast \leq \bar{R}_{\rm P2}^\ast$.

To prove b), we plug the stationary and randomized policy $\Gamma$ that satisfies the Slater conditions \eqref{SLT} into the RHS of the inequality $(\dag)$ in \eqref{Lem_eq1}. By removing the negative term $-\epsilon\lambda_{\rm e}\left(\Omega-B_t\right)$, we obtain that
\begin{equation}
\Delta^t-V\mathbb{E}\left[r_t^{\Psi}|\Theta_t\right]\leq D-V\varphi(\epsilon)-\left(Z_t\lambda_{\rm c}+Q_t^{\rm S}+Q_t^{\rm U}\right)\epsilon.
\end{equation}
Taking iterated expectations and telescoping sums, and normalizing by $N\epsilon$, we have that
\begin{equation}
\begin{split}
&\frac{\mathbb{E}\left[\Phi_{t}\right]\!-\!\mathbb{E}\left[\Phi_{0}\right]}{N\epsilon}\!-\!\frac{V}{N\epsilon}\mathsmaller{\sum}_{t=0}^{N-1}\mathbb{E}\left[r_t^{\Psi}\right]
\leq\! \frac{D-V\varphi(\epsilon)}{\epsilon}-\frac{1}{N}\mathsmaller{\sum}_{t=0}^{N}\mathbb{E}\left[Z_t\lambda_{\rm c}+Q_t^{\rm S}+Q_t^{\rm U}\right].
\end{split}
\end{equation}
Letting $N\to\infty$, rearranging the terms and using the fact that $\lim_{N\!\to\!+\!\infty}\frac{1}{N}\sum_{t=0}^{N-1}\mathbb{E}\left[r_t^{\Psi}\right]\leq \bar{R}_{\rm P1}^\ast$, we have
\begin{equation}\label{app2_eq1}
\lim_{N\!\to\!+\!\infty}\frac{1}{N}\mathsmaller{\sum}_{t=0}^{N}\mathbb{E}\left[Z_t\lambda_{\rm c}+Q_t^{\rm S}+Q_t^{\rm U}\right] \leq \frac{D+V\left(\bar{R}_{\rm P1}^\ast-\varphi(\epsilon)\right)}{\epsilon},
\end{equation}
which implies the strong stability of $Q_t^{\rm U}$, $Q_t^{\rm S}$ and $Z_t$, i.e., 
\begin{equation}
\lim_{N\!\to\!+\!\infty}\frac{1}{N}\mathsmaller{\sum}_{t=0}^{N}\mathbb{E}\left[Q_t^{\rm U}\right]<\infty,\lim_{N\!\to\!+\!\infty}\frac{1}{N}\mathsmaller{\sum}_{t=0}^{N}\mathbb{E}\left[Q_t^{\rm S}\right]<\infty,\lim_{N\!\to\!+\!\infty}\frac{1}{N}\mathsmaller{\sum}_{t=0}^{N}\mathbb{E}\left[Z_t\right]<\infty.
\end{equation}
Since $Z_t$, $Q_t^{\rm U}$ and $Q_t^{\rm S}$ are non-negative, we obtain the results in \eqref{lem5_eq2}. Meanwhile, because strong stability implies mean rate stable (see Theorem 2.8 in \cite{Neely2010}), the long-term average power constraint \eqref{Bud_cons} is satisfied according to Lemma \ref{lem7}, which proves c).

\ifCLASSOPTIONcaptionsoff
  \newpage
\fi

\bibliographystyle{IEEEtran}
\bibliography{IEEEabrv,MyRefLib}

\end{document}